%% file: outline.tex
\documentclass{sig-alternate}

\usepackage{enumitem}
\usepackage{framed}
\usepackage{cprotect}
\usepackage{enumitem}
\usepackage{listings}
\usepackage{amstext}
\usepackage{amstext}
\usepackage{pdfpages}
\usepackage{alltt}
\usepackage{epstopdf}
\usepackage{xspace,colortbl}
\usepackage[USenglish]{babel}
\usepackage{multirow}
\usepackage{url}
\usepackage{subfigure}
\usepackage{graphicx}
\usepackage{amssymb}
\usepackage{fmtcount}
\usepackage{amsfonts}
\usepackage{xspace}
\usepackage{amsmath}
\usepackage{multirow}
\usepackage[mathscr]{eucal}
\usepackage{colortbl}

\usepackage{bm}
\usepackage{charter}
\usepackage{csquotes}
\usepackage{enumitem}

\makeatletter
\def\@copyrightspace{\relax}
\makeatother

\begin{document}


\newtheorem{theorem}{Theorem}
\newtheorem{example}{Example}
\newtheorem{definition}{Definition}
\newtheorem{problem}{Problem}
\newtheorem{property}{Property}
\newtheorem{proposition}{Proposition}
\newtheorem{lemma}{Lemma}
\newtheorem{corollary}{Corollary}

\newcommand{\cond}{\textrm{pred}\xspace}
\newcommand{\dataset}{data set\xspace}
\newcommand{\datasets}{data sets\xspace}
\newcommand{\spview}{\textsf{SPView}\xspace}
\newcommand{\fjview}{\textsf{FJView}\xspace}
\newcommand{\aggview}{\textsf{AggView}\xspace}
\newcommand{\hashfunc}[1]{\textsf{hash}(#1)\xspace}
\newcommand{\hashop}{\textsf{hash}\xspace}
\newcommand{\nsc}{\textsf{NormalizedSC}\xspace}
\newcommand{\rsc}{\textsf{RawSC}\xspace}

\newcommand{\avgfunc}{\ensuremath{\texttt{avg} }\xspace}
\newcommand{\maxfunc}{\ensuremath{\texttt{max} }\xspace}
\newcommand{\minfunc}{\ensuremath{\texttt{min} }\xspace}
\newcommand{\histfunc}{\ensuremath{\texttt{histogram\_numeric} }\xspace}
\newcommand{\countfunc}{\ensuremath{\texttt{count}}\xspace}
\newcommand{\sumfunc}{\ensuremath{\texttt{sum} }\xspace}
\newcommand{\varfunc}{\ensuremath{\texttt{var} }\xspace}
\newcommand{\stdfunc}{\ensuremath{\texttt{std} }\xspace}
\newcommand{\covfunc}{\ensuremath{\texttt{cov} }\xspace}
\newcommand{\corrfunc}{\ensuremath{\texttt{corr} }\xspace}
\newcommand{\medfunc}{\ensuremath{\texttt{median} }\xspace}
\newcommand{\percfunc}{\ensuremath{\texttt{percentile} }\xspace}
\newcommand{\havingfunc}{\ensuremath{\texttt{HAVING} }\xspace}
\newcommand{\selectfunc}{\ensuremath{\texttt{select} }\xspace}
\newcommand{\ratio}{\ensuremath{\rho }\xspace}

\newcommand{\insertion}{\ensuremath{\texttt{INSERT} }\xspace}
\newcommand{\update}{\ensuremath{\texttt{UPDATE} }\xspace}
\newcommand{\delete}{\ensuremath{\texttt{DELETE} }\xspace}

\newcommand{\sysfull}{ActiveClean\xspace}
\newcommand{\sys}{ActiveClean\xspace}
\newcommand{\sysnospace}{ActiveClean}

\newcommand{\tbl}[1]{\textsf{#1}\xspace}
\newcommand{\field}[1]{\textsf{#1}\xspace}
\newcommand{\cost}{\textrm{cost}\xspace}
\newcommand{\ans}{\textsf{ans}\xspace}
\newcommand{\dans}{\Delta\textsf{ans}\xspace}
\newcommand{\cqp}{correction query processing\xspace}
\newcommand{\Cqp}{Correction query processing\xspace}

\newcommand{\reminder}[1]{{{\textcolor{magenta}{\{\{\bf #1\}\}}}\xspace}}
\newcommand{\specialcell}[2][c]{%
  \begin{tabular}[#1]{@{}c@{}}#2\end{tabular}}

\def\ojoin{\setbox0=\hbox{$\bowtie$}%
  \rule[-.02ex]{.25em}{.4pt}\llap{\rule[\ht0]{.25em}{.4pt}}}
\def\leftouterjoin{\mathbin{\ojoin\mkern-5.8mu\bowtie}}
\def\rightouterjoin{\mathbin{\bowtie\mkern-5.8mu\ojoin}}
\def\fullouterjoin{\mathbin{\ojoin\mkern-5.8mu\bowtie\mkern-5.8mu\ojoin}}

\pagestyle{plain}

\title{ActiveClean: Interactive Data Cleaning While Learning Convex Loss Models}

\numberofauthors{1}
\author{\large Sanjay Krishnan, Jiannan Wang, Eugene Wu{$\,^\dag$}, Michael J. Franklin, Ken Goldberg \\
\vspace{.2em}\affaddr{\large UC Berkeley, ~~ $^\dag$Columbia University} \\
\vspace{.1em}\affaddr{\large \{sanjaykrishnan, jnwang, franklin, goldberg\}@berkeley.edu}\\
\affaddr{\large ewu@cs.columbia.edu}
}


\maketitle

\begin{abstract}
Data cleaning is often an important step to ensure that predictive models, such as regression and classification, are not affected by systematic errors such as inconsistent, out-of-date, or outlier data.
Identifying dirty data is often a manual and iterative process, and can be challenging on large datasets.
However, many data cleaning workflows can introduce subtle biases into the training processes due to violation of independence assumptions.
We propose \sys, a progressive cleaning approach where the model is updated incrementally instead of re-training and can guarantee accuracy on partially cleaned data.
\sys supports a popular class of models called convex loss models (e.g., linear regression and SVMs).
\sys also leverages the structure of a user's model to prioritize cleaning those records likely to affect the results.
We evaluate \sys on five real-world datasets UCI Adult, UCI EEG, MNIST, Dollars For Docs, and WorldBank with both real and synthetic errors.
Our results suggest that our proposed optimizations can improve model accuracy by up-to 2.5x for the same amount of data cleaned.
Furthermore for a fixed cleaning budget and on all real dirty datasets, \sys returns more accurate models than uniform sampling and Active Learning. 
\end{abstract}

\if{0}
\begin{abstract}
Databases are susceptible to various forms of corruption, or \emph{dirtiness}, such as missing, incorrect, or inconsistent values.
Increasingly, modern data analysis pipelines involve Machine Learning for predictive models which can be sensitive to dirty data.
Dirty data is often expensive to repair, and naive sampling solutions are not suited for training high dimensional models.
In this paper, we propose \sysfull, an anytime framework for training Machine Learning models with budgeted data cleaning.
Our framework updates a model iteratively as small samples of data are cleaned, and includes numerous optimizations such as importance weighting and dirty data detection.
We evaluate \sys on 4 real datasets and find that our methodology can return more accurate models for a smaller cost  than alternatives such as uniform sampling and active learning.
\end{abstract}
\fi

\setcounter{page}{1}

\input{introduction.tex}

\input{background.tex}

\input{problem_statement.tex}

\input{architecture.tex}

\input{naive.tex}
\input{sampling.tex}
\input{optimal.tex}

\input{experiments.tex}
\input{relatedwork.tex}

\input{discussion.tex}

\input{conclusion.tex}

\textbf{\small This research is supported in part by NSF CISE Expeditions Award CCF-1139158, LBNL Award 7076018, and DARPA XData Award FA8750-12-2-0331, and gifts from Amazon Web Services, Google, SAP, The Thomas and Stacey Siebel Foundation, Adatao, Adobe, Apple, Inc., Blue Goji, Bosch, C3Energy, Cisco, Cray, Cloudera, EMC2, Ericsson, Facebook, Guavus, HP, Huawei, Informatica, Intel, Microsoft, NetApp, Pivotal, Samsung, Schlumberger, Splunk, Virdata and VMware.}

\fontsize{7.5pt}{8.0pt} \selectfont
\bibliographystyle{abbrv}
\bibliography{ref} 
\normalsize \selectfont
\appendix
\input{appendix.tex}

\end{document}

%% file: introduction.tex
\section{Introduction}
Machine Learning on large and growing datasets is a key data management challenge with significant interest in both industry and academia~\cite{bdas, alexandrov2014stratosphere, crotty2014tupleware, tensor}.
Despite a number of breakthroughs in reducing training time, predictive modeling can still be a tedious and time-consuming task for an analyst. 
Data often arrive \emph{dirty}, including missing, incorrect, or inconsistent attributes, and analysts widely report that data cleaning and other forms of pre-processing account for up to 80\% of their effort~\cite{nytimes, kandel2012}.
While data cleaning is an extensively studied problem, the predictive modeling setting poses a number of new challenges: (1) high dimensionality can amplify even a small amount of erroneous records~\cite{xiaofeature}, (2) the complexity can make it difficult to trace the consequnces of an error, and (3) there are often subtle technical cconditions (e.g., independent and identically distributed) that can be violated by data cleaning.
Consequently, techniques that have been designed for traditional SQL analytics may be inefficient or even unreliable.
In this paper, we study the relationship between data cleaning and model training workflows and explore how to apply existing data cleaning approaches with provable guarantees.

One of the main bottlenecks in data cleaning is the human effort in determining which data are dirty and then developing rules or software to correct the problems.
For some types of dirty data, such as inconsistent values, model training may seemingly succeed albeit, with potential subtle inaccuracies in the model.
For example, battery-powered sensors can transmit unreliable measurements when battery levels are low \cite{DBLP:conf/pervasive/JefferyAFHW06}. 
Similarly, data entered by humans can be susceptible to a variety of inconsistencies (e.g., typos), and unintentional cognitive biases~\cite{DBLP:conf/recsys/KrishnanPFG14}.
Such problems are often addressed in time-consuming loop where the analys trains a model, inspects the model and its predictions, clean some data, and re-train.

This iterative process is the de facto standard, but without appropriate care, can lead to several serious statistical issues.
Due to the well-known Simpson's paradox, models trained on a mix of dirty and clean data can have very misleading results even in simple scenarios (Figure \ref{update-arch1}).
Furthermore, if the candidate dirty records are not identified with a known sampling distribution, the statistical independence assumptions for most training methods are violated. 
The violations of these assumptions can introduce confounding biases.
To this end, we designed \sys which trains predictive models while allowing for iterative data cleaning and has accuracy guarantees.
\sys automates the dirty data identification process and the model update process, thereby abstracting these two error-prone steps away from the analyst.

\sys is inspired by the recent success of progressive data cleaning where a user can gradually clean more data until the desired accuracy is reached~\cite{altowim2014progressive, whang2014incremental, papenbrock2015progressive, gruenheid2014incremental, mayfield2010eracer, DBLP:journals/pvldb/YakoutENOI11, yakout2013don}.
We focus on a popular class of models called convex loss models (e.g., includes linear regression and SVMs) and show that the Simpson's paradox problem can be avoided using iterative maintenance of a model rather than re-training.
This process leverages the convex structure of the model rather than treating it like a black-box, and we apply convergence arguments from convex optimization theory.
We propose several novel optimizations that leverage information from the model to guide data cleaning towards the records most likely to be dirty and most likely to affect the results.

\noindent To summarize the contributions:
\begin{itemize}[noitemsep]
\item \textbf{Correctness} (Section \ref{model-update}). We show how to update a dirty model given newly cleaned data. This update converges monotonically in expectation. For a batch size $b$ and iterations $T$, it converges with rate $O(\frac{1}{\sqrt{bT}})$. 
\item \textbf{Efficiency} (Section \ref{dist-samp}). We derive a theoretical optimal sampling distribution that minimizes the update error and an approximation to estimate the theoretical optimum.
\item \textbf{Detection and Estimation} (Section \ref{opti}). We show how \sys can be integrated with data detection to guide data cleaning towards records expected to be dirty.
\item The experiments evaluate these components on four datasets with real and synthetic corruption (Section \ref{eval}). Results suggests that for a fixed cleaning budget, \sys returns more accurate models than uniform sampling and Active Learning when systematic corruption is sparse.

\end{itemize}

%% file: background.tex
\section{Background and Problem Setup}\label{background}
This section formalizes the iterative data cleaning and training process and highlights an example application.

\subsection{Predictive Modeling}
The user provides a relation $R$ and wishes to train a model using the data in $R$.
This work focuses on a class of well-analyzed predictive analytics problems; ones that can be expressed as the minimization of convex loss functions.
Convex loss minimization problems are amenable to a variety of incremental optimization methodologies with provable guarantees (see Friedman, Hastie, and Tibshirani \cite{friedman2001elements} for an introduction).
Examples include generalized linear models (including linear and logistic regression), support vector machines, and in fact, means and medians are also special cases. 

We assume that the user provides a featurizer $F(\cdot)$ that maps every record $r \in R$ to a feature vector $x$ and label $y$.
For labeled training examples $\{(x_{i},y_{i})\}_{i=1}^{N}$, the problem is to find a vector of \emph{model parameters} $\theta$ by minimizing a loss function $\phi$ over all training examples:
\[
 \theta^{*}=\arg\min_{\theta}\sum_{i=1}^{N}\phi(x_{i},y_{i},\theta)
\]
Where $\phi$ is a convex function in $\theta$.
For example, in a linear regression $\phi$ is:
\[
\phi(x_{i},y_{i},\theta) = \|\theta^Tx_{i} - y_i \|_2^2
\]
Typically, a \emph{regularization} term $r(\theta)$ is added to this problem.
$r(\theta)$ penalizes high or low values of feature weights in $\theta$ to avoid overfitting to noise in the training examples.
\begin{equation}
 \theta^{*}=\arg\min_{\theta}\sum_{i=1}^{N}\phi(x_{i},y_{i},\theta) + r(\theta)
 \label{ideal}
\end{equation}
In this work, without loss of generality, we will include the regularization as part of the loss function i.e., $\phi(x_{i},y_{i},\theta)$ includes $r(\theta)$.

\subsection{Data Cleaning}
We consider corruption that affects the attribute values of records. This does \emph{not} cover errors that simultaneously affect multiple records such as record duplication or structure such as schema transformation.
Examples of supported cleaning operations include, batch resolving common inconsistencies (e.g., merging ``U.S.A" and ``United States"), filtering outliers (e.g., removing records with values $>1e6$), and standardizing attribute semantics (e.g., ``1.2 miles" and ``1.93 km").

We are particularly interested in those errors that are difficult or time-consuming to clean, and require the analyst to examine an erroneous record, and determine the appropriate action--possibly leveraging knowledge of the current best model.
We represent this operation as $Clean(\cdot)$ which can be applied to a record $r$ (or a set of records) to recover the clean record $r' = Clean(r)$.
Formally, we treat the $Clean(\cdot)$ as an expensive user-defined function composed of deterministic schema-preserving \textsf{map} and \textsf{filter} operations applied to a subset of rows in the relation.
A relation is defined as \emph{clean} if $R_{clean} = Clean(R_{clean})$.
Therefore, for every $r \in R_{clean}$ there exists a unique $r' \in R$ in the dirty data.
The \textsf{map} and \textsf{filter} cleaning model is not a fundamental restriction of \sys, and Appendix~\ref{set-of-r} discusses a compatible ``set of records" cleaning model.

\subsection{Iteration}
As an example of how $Clean(\cdot)$ fits into an iterative analysis process, consider an analyst training a regression and identifying outliers. 
When she examines one of the outliers, she realizes that the base data (prior to featurization) has a formatting inconsistency that leads to incorrect parsing of the numerical values.
She applies a batch fix (i.e., $Clean(\cdot)$) to all of the outliers with the same error, and re-trains the model. 
This iterative process can be described as the following pseudocode loop:
\begin{enumerate}[leftmargin=1em]\scriptsize\sloppy
  \item \texttt{Init(iter)}
  \item \texttt{current\_model = Train(R)}
  \item For each t in $\{1,...,iter\}$
  \begin{enumerate}
    \item \texttt{dirty\_sample $=$ Identify(R,current\_model)}
    \item \texttt{clean\_sample $=$ Clean(dirty\_sample)}
    \item \texttt{current\_model $=$ Update(clean\_sample, R)}
  \end{enumerate}
  \item \texttt{Output: current\_model}
  \end{enumerate}

\vspace{0.5em}

While we have already discussed $Train(\cdot)$ and $Clean(\cdot)$, the analyst still has to define the primitives $Identify(\cdot)$ and $Update(\cdot)$.
For $Identify(\cdot)$, given a the current best model, the analyst must specify some criteria to select a set of records to examine.
And in $Update(\cdot)$, the analyst must decide how to update the model given newly cleaned data.
It turns out that these primitives are not trivial to implement since the straight-forward solutions can actually lead to divergence of the trained models.

\subsection{Challenges}\label{correctness} 
\vspace{0.5em} 
\textbf{Correctness: } Let us assume that the analyst has implemented an $Identify(\cdot)$ function that returns $k$ candidate dirty records.
The straight-forward application data cleaning is to repair the corruption in place, and re-train the model after each repair.
Suppose $k \ll N$ records are cleaned, but all of the remaining dirty records are retained in the dataset.
Figure \ref{update-arch1} highlights the dangers of this approach on a very simple dirty dataset and a linear regression model i.e., the best fit line for two variables. 
One of the variables is systematically corrupted with a translation in the x-axis (Figure \ref{update-arch1}a).
The dirty data is marked in red and the clean data in blue, and they are shown with their respective best fit lines.
After cleaning only two of the data points (Figure \ref{update-arch1}b), the resulting best fit line is in the opposite direction of the true model.

Aggregates over mixtures of different populations of data can result in spurious relationships due to the well-known phenomenon called Simpson's paradox \cite{simpson1951interpretation}.
Simpson's paradox is by no means a corner case, and it has affected the validity of a number of high-profile studies~\cite{simpsonsparadox}; even in the simple case of taking an average over a dataset.
Predictive models are high-dimensional generalizations of these aggregates without closed form techniques to compensate for these biases.
Thus, training models on a mixture of dirty and clean data can lead to unreliable results, where artificial trends introduced by the mixture can be confused for the effects of data cleaning.

\begin{figure}[ht!]
\centering
 \includegraphics[width=\columnwidth]{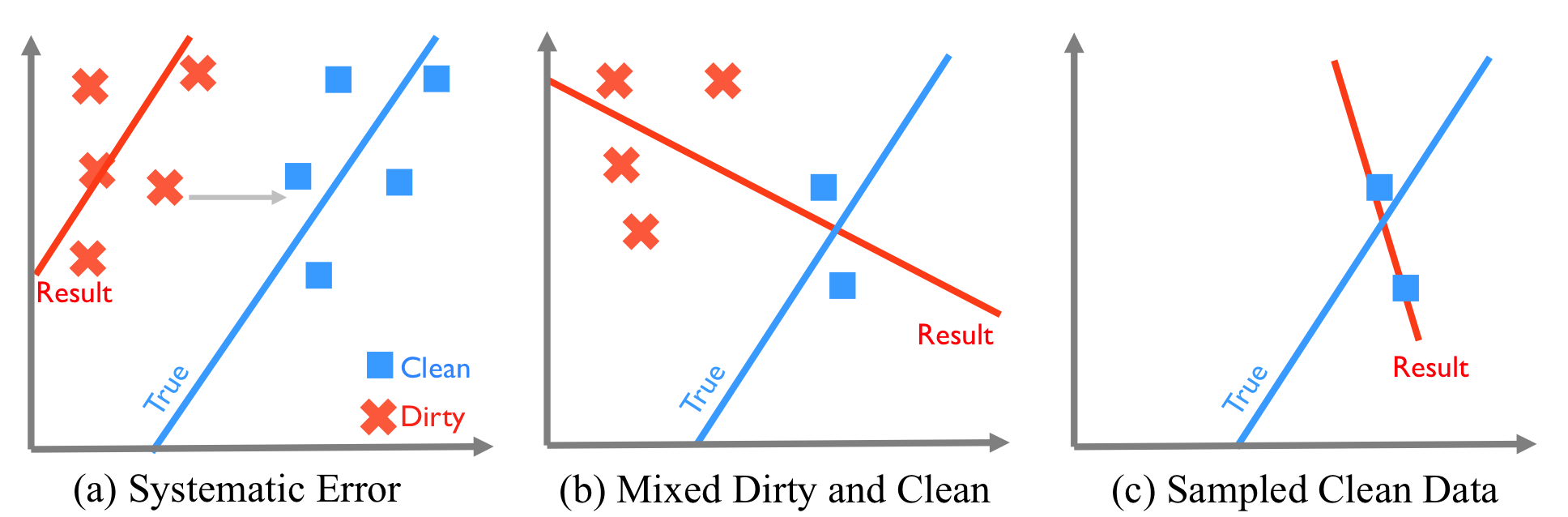}
 \caption{(a) Systematic corruption in one variable can lead to a shifted model. 
 (b) Mixed dirty and clean data results in a less accurate model than no cleaning.
(c) Small samples of only clean data can result in similarly inaccurate models. \label{update-arch1}}
\end{figure}

An alternative is to avoid the dirty data altogether instead of mixing the two populations, and the model re-training is restricted to only data that are known to be clean.
This approach is similar to SampleClean \cite{wang1999sample}, which was proposed to approximate the results of aggregate queries by applying them to a clean sample of data.
However, high-dimensional models are highly sensitive to sample size.
Figure \ref{update-arch1}c illustrates that, even in two dimensions, models trained from small samples can be as incorrect as the mixing solution described before.

\vspace{0.5em} 

\textbf{Efficiency: } Conversely, hypothetically assume that the analyst has implemented a correct $Update(\cdot)$ primitive and implements $Identify(\cdot)$ with a technique such as Active Learning to select records to clean~\cite{yakout2013don,DBLP:journals/pvldb/YakoutENOI11,gokhale2014corleone}.
Active learning is a technique to carefully select the set of examples to learn the most accurate model.
However, these selection criteria are designed for stationary data distributions, an assumption which is not true in this setting.
As more data are cleaned, the data distribution changes.
Data which may look unimportant in the dirty data might be very valuable to clean in reality, and thus any prioritization has to predict a record's value with respect to an anticipated clean model.

\subsection{The Need For Automation}\label{alrw}
\sys is a framework that implements the $Identify(\cdot)$ and $Update(\cdot)$ primitives for the analyst. 
By automating the iterative process, \sys ensures reliable models with convergence guarantees.
The analyst first initializes \sys with a dirty model.
\sys carefuly selects small batches of data to clean based on data that are likely to be dirty and likely to affect the model.
The analyst applies data cleaning to these batches, and \sys updates the model with an incremental optimization technique.

Machine learning has been applied in prior work to improve the efficiency of data cleaning~\cite{yakout2013don,DBLP:journals/pvldb/YakoutENOI11,gokhale2014corleone}.
Human input, either for cleaning or validation of automated cleaning, is often expensive and impractical for large datasets.
A model can learn rules from a small set of examples cleaned (or validated) by a human, and active learning is a technique to carefully select the set of examples to learn the most accurate model.
This model can be used to extrapolate repairs to not-yet-cleaned data, and the goal of these approaches is to provide the cleanest possible dataset--independent of the subsequent analytics or query processing.
These approaches, while very effective, suffer from composibility problems when placed inside cleaning and training loops.
To summarize, \sys considers data cleaning \emph{during} model training, while these techniques consider model training \emph{for} data cleaning.
One of the primary contributions of this work is an incremental model update algorithm with correctness guarantees for mixtures of data.

\subsection{Use Case: Dollars for Docs \cite{dollarsfordocs}}\label{s:usecase}
ProPublica collected a dataset of corporate donations to doctors to analyze conflicts of interest. 
They reported that some doctors received over \$500,000 in travel, meals, and consultation expenses \cite{dollarsfordocsa}.
ProPublica laboriously curated and cleaned a dataset from the Centers for Medicare and Medicaid Services that listed nearly 250,000 research donations, and aggregated these donations by physician, drug, and pharmaceutical company.
We collected the raw unaggregated data and explored whether suspect donations could be predicted with a model.
This problem is typical of analysis scenarios based on observational data seen in finance, insurance, medicine, and investigative journalism.
The dataset has the following schema:
\begin{lstlisting}[mathescape,basicstyle={\scriptsize}]
Contribution(pi_specialty$\textrm{,}$ drug_name$\textrm{,}$ device_name$\textrm{,}$
corporation$\textrm{,}$ amount$\textrm{,}$ dispute$\textrm{,}$ status)
\end{lstlisting}

\noindent\texttt{pi\_specialty} is a textual attribute describing the specialty of the doctor receiving the donation.

\noindent\texttt{drug\_name} is the branded name of the drug in the research study (null if not a drug).

\noindent\texttt{device\_name} is the branded name of the device in the study (null if not a device).

\noindent\texttt{corporation} is the name of the pharmaceutical providing the donation.

\noindent\texttt{amount} is a numerical attribute representing the donation amount.

\noindent\texttt{dispute} is a Boolean attribute describing whether the research was disputed.

\noindent\texttt{status} is a string label describing whether the  donation was allowed under the declared research protocol. The goal is to predict disallowed  donation. 

\vspace{0.5em}

However, this dataset is very dirty, and the systematic nature of the data corruption can result in an inaccurate model.
On the ProPublica website \cite{dollarsfordocs}, they list numerous types of data problems that had to be cleaned before publishing the data (see Appendix \ref{dfd-errors}).
For example, the most significant donations were made by large companies whose names were also more often inconsistently represented in the data, e.g., ``Pfizer Inc.", ``Pfizer Incorporated", ``Pfizer".
In such scenarios, the effect of systematic error can be serious.
Duplicate representations could artificially reduce the correlation between these entities and suspected contributions.
There were nearly 40,000 of the 250,000 records that had either naming inconsistencies or other inconsistencies in labeling the allowed or disallowed \texttt{status}.
Without data cleaning, the detection rate using a Support Vector Machine was 66\%.
Applying the data cleaning to the entire dataset improved this rate to 97\% in the clean data (Section \ref{dfd-exp}), and the experiments describe how \sys can achieve an 80\% detection rate for less than 1.6\% of the records cleaned.

%% file: problem_statement.tex
\section{Problem Formalization}\label{statements}
\noindent This section formalizes the problems addressed in the paper.

\subsection{Notation and Setup}\label{notation}
The user provides a relation $R$, a cleaner $C(\cdot)$, a featurizer $F(\cdot)$, and a convex loss problem defined by the loss $\phi(\cdot)$.
A total of $k$ records will be cleaned in batches of size $b$, so there will be $\frac{k}{b}$ iterations. 
We use the following notation to represent relevant intermediate states:
\begin{itemize}[noitemsep]
\item \textbf{Dirty Model: } $\theta^{(d)}$ is the model trained on $R$ (without cleaning) with the featurizer $F(\cdot)$ and loss $\phi(\cdot)$. This serves as an initialization to \sys.
\item \textbf{Dirty Records: } $R_{dirty} \subseteq R$ is the subset of records that are still dirty. As more data are cleaned $R_{dirty} \rightarrow \{\}$.
\item \textbf{Clean Records: } $R_{clean} \subseteq R$ is the subset of records that are clean, i.e., the complement of $R_{dirty}$.
\item \textbf{Samples: } $S$ is a sample (possibly non-uniform but with known probabilities) of the records $R_{dirty}$. The clean sample is denoted by $S_{clean} = C(S)$.
\item \textbf{Clean Model: } $\theta^{(c)}$ is the optimal clean model, i.e., the model trained on a fully cleaned relation.
\item \textbf{Current Model: } $\theta^{(t)}$ is the current best model at iteration $t \in \{1,...,\frac{k}{b}\}$, and $\theta^{(0)} = \theta^{(d)}$. 
\end{itemize}

There are two metrics that we will use to measure the performance of \sys:

\vspace{0.25em}

\noindent\textbf{Model Error. } The model error is defined as $\|\theta^{(t)} - \theta^{(c)}\|$.

\vspace{0.25em}

\noindent\textbf{Testing Error. } Let $T(\theta^{(t)})$ be the out-of-sample testing error when the current best model is applied to the clean data, and $T(\theta^{(c)})$ be the test error when the clean model is applied to the clean data. The testing error is defined as $T(\theta^{(t)}) - T(\theta^{(c)})$

\subsection{Problem 1. Correct Update Problem}\label{updp}
Given newly cleaned data $S_{clean}$ and the current best model $\theta^{(t)}$, the model update problem is to calculate $\theta^{(t+1)}$. 
$\theta^{(t+1)}$ will have some error with respect to the true model $\theta^{(c)}$, which we denote as:
\[
error(\theta^{(t+1)}) = \| \theta^{(t+1)} - \theta^{(c)} \|
\]
Since a sample of data are cleaned, it is only meaningful to talk about expected errors.
We call the update algorithm ``reliable" if the expected error is upper bounded by a monotonically decreasing function $\mu$ of the amount of cleaned data:
\[
\mathbb{E}(error(\theta^{new})) = O(\mu(\mid S_{clean} \mid))
\]
Intuitively, ``reliable" means that more cleaning should imply more accuracy.

\vspace{0.5em}

\emph{The Correct Update Problem is to reliably update the model $\theta^{(t)}$ with a sample of cleaned data.}

\subsection{Problem 2. Efficiency Problem}\label{optp}
The efficiency problem is to select $S_{clean}$ such that the expected error $\mathbb{E}(error(\theta^{(t)}))$ is minimized.
\sys uses previously cleaned data to estimate the value of data cleaning on new records.
Then it draws a sample of records $S \subseteq R_{dirty}$. 
This is a non-uniform sample where each record $r$ has a sampling probability $p(r)$ based on the estimates.
We derive the optimal sampling distribution for the SGD updates, and show how the theoretical optimum can be approximated.

\vspace{0.5em}

\emph{The Efficiency Problem is to select a sampling distribution $p(\cdot)$ over all records such that the expected error w.r.t to the model if trained on fully clean data is minimized.}

%% file: architecture.tex
\section{Architecture}\label{arch}
\noindent This section presents the \sys architecture.

\subsection{Overview}\label{sysover}
Figure \ref{sys-arch} illustrates the \sys architecture.
The dotted boxes describe optional components that the user can provide to improve the efficiency of the system.  

\subsubsection{Required User Input}\label{uinp}

\noindent\textbf{Model:} The user provides a predictive model (e.g., SVM) specified as a convex loss optimization problem $\phi(\cdot)$ and a featurizer $F(\cdot)$ that maps a record to its feature vector $x$ and label $y$.

\vspace{0.25em}

\noindent\textbf{Cleaning Function: } The user provides a function $C(\cdot)$ (implemented via software or crowdsourcing) that maps dirty records to clean records as per our definition in Section \ref{dmodel}. 

\vspace{0.25em}

\noindent\textbf{Batches: } Data are cleaned in batches of size $b$ and the user can change these settings if she desires more or less frequent model updates.
The choice of $b$ does affect the convergence rate.
Section~\ref{model-update} discusses the efficiency and convergence trade-offs of different values of $b$.
We empirically find that a batch size of $50$ performs well across different datasets and use that as a default.
A cleaning budget $k$ can be used as a stopping criterion once $C(\dot)$ has been called $k$ times, and so the number of iterations of \sys is $T = \frac{k}{b}$.
Alternatively, the user can clean data until the model is of sufficient accuracy to make a decision.

\subsubsection{Basic Data Flow} \label{df}
The system first trains the model $\phi(\cdot)$ on the dirty dataset to find an initial model $\theta^{(d)}$ that the system will subsequently improve.
The {\it sampler} selects a sample of size $b$ records from the dataset and passes
the sample to the {\it cleaner}, which executes $C(\cdot)$ for each sample record and outputs their cleaned versions.
The \emph{updater} uses the cleaned sample to update the weights of the model, thus moving the model closer to the true cleaned model (in expectation).
Finally, the system either terminates due to a stopping condition (e.g., $C(\cdot)$ has been called a maximum number of times $k$, or training error convergence),
or passes control to the {\it sampler} for the next iteration.

\subsubsection{Optimizations}
In many cases, such as missing values, errors can be efficiently detected.
A user provided {\it Detector} can be used to identify such records that are more likely to be dirty, and thus improves the likelihood that the next sample will contain true dirty records.
Furthermore, the {\it Estimator} uses previously cleaned data to estimate the effect that cleaning a given record will have on the model.
These components can be used separately (if only one is supplied) or together to focus the system's cleaning efforts on records that will most improve the model.
Section \ref{opti} describes several instantiations of these components for different data cleaning problems.
Our experiments show that these optimizations can improve model accuracy by up-to 2.5x (Section \ref{comp}).

\subsection{Example}
The following example illustrates how a user would apply \sys to address the use case in Section~\ref{s:usecase}:
\begin{example}\label{archex}
The analyst chooses to use an SVM model, and manually cleans records by hand (the $C(\cdot)$).  
\sys initially selects a sample of $50$ records (the default)  to show the analyst.
She identifies a subset of 15 records that are dirty, fixes them by normalizing the drug and corporation names with the help of a search engine, and corrects the labels with typographical or incorrect values.
The system then uses the cleaned records to update the the current best model and select the next sample of $50$.
The analyst can stop at any time and use the improved model to predict donation likelihoods.
\end{example}

\begin{figure}[t]
\centering
 \includegraphics[width=\columnwidth]{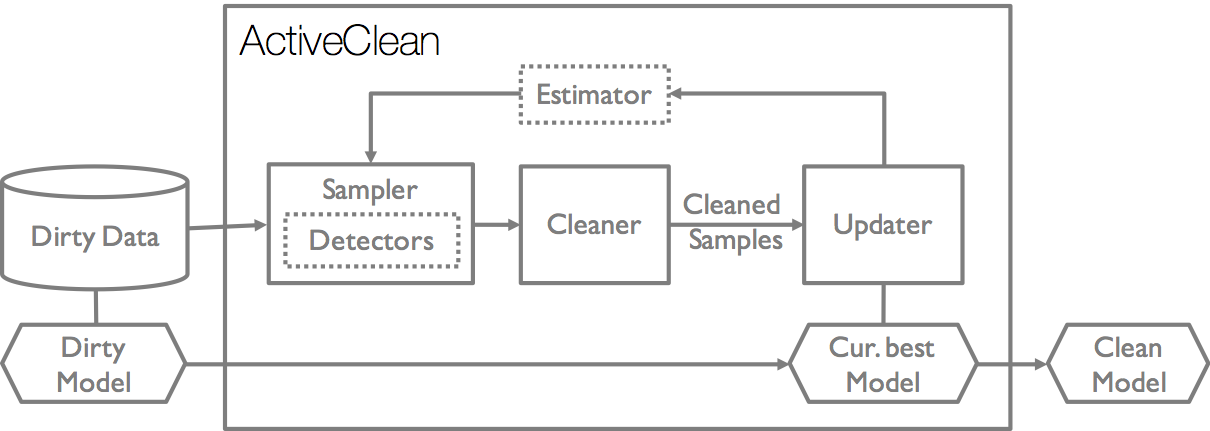}
 \caption{\sysfull allows users to train predictive models while progressively cleaning data. The framework adaptively selects the best data to clean and can optionally (denoted with dotted lines) integrate with pre-defined detection rules and estimation algorithms for improved conference. \label{sys-arch}}
\end{figure}

%% file: naive.tex
\section{Updates With Correctness}\label{model-update}
This section describes an algorithm for reliable model updates.
The updater assumes that it is given a sample of data $S_{dirty}$ from $R_{dirty}$ where $i \in S_{dirty}$ has a known sampling probability $p(i)$.
Sections \ref{dist-samp} and \ref{opti} show how to optimize $p(\cdot)$ and the analysis in this section applies for any sampling distribution $p(\cdot) > 0$.

\subsection{Geometric Derivation}\label{geod}
The update algorithm intuitively follows from the convex geometry of the problem.
Consider the problem in one dimension (i.e., the parameter $\theta$ is a scalar value), so then the goal is to find the minimum point ($\theta$) of a curve $l(\theta)$.
The consequence of dirty data is that the wrong loss function is optimized.
Figure \ref{update-arch2}A illustrates the consequence of the optimization.
The red dotted line shows the loss function on the dirty data.
Optimizing the loss function finds $\theta^{(d)}$ that at the minimum point (red star).
However, the true loss function (w.r.t to the clean data) is in blue, thus
the optimal value on the dirty data is in fact a suboptimal point on clean curve (red circle).

\begin{figure}[ht!]
\centering
 \includegraphics[width=\columnwidth]{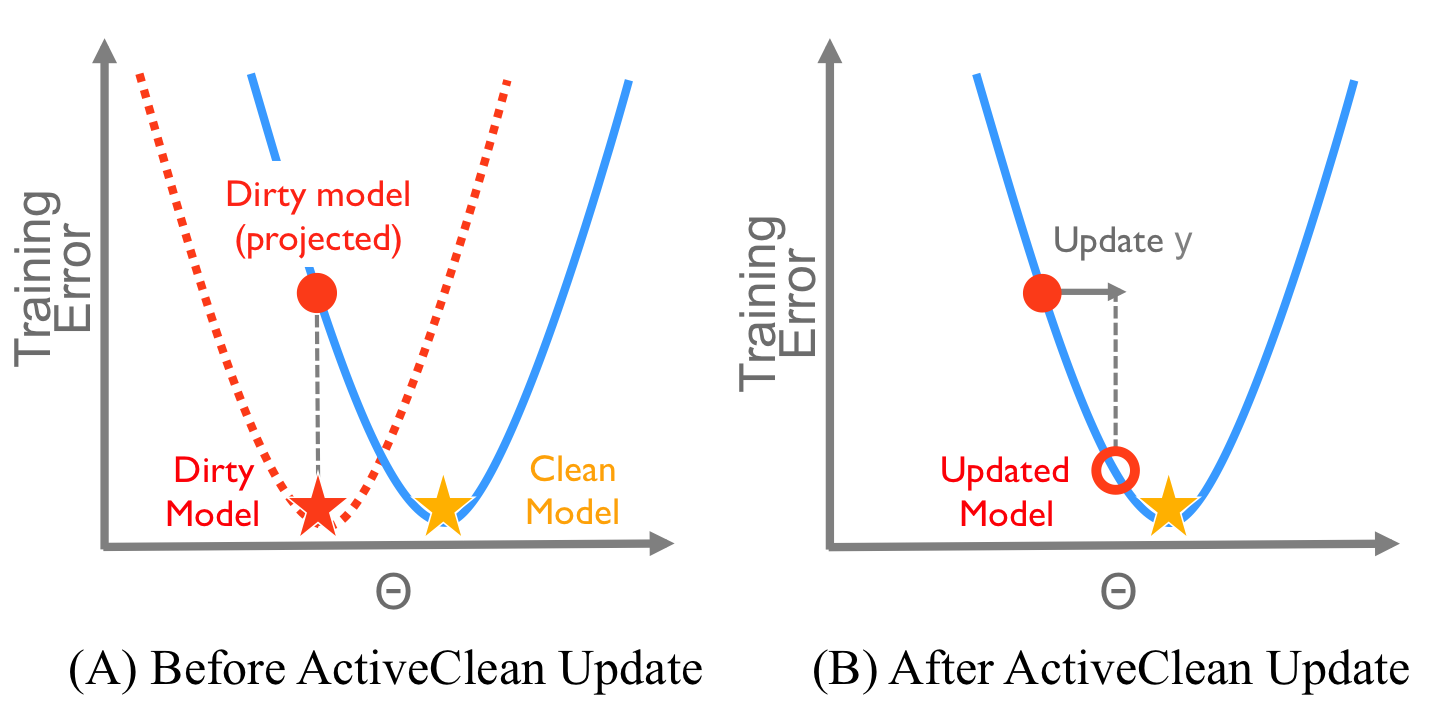}\vspace{-1em}
 \caption{(A) A model trained on dirty data can be thought of as a sub-optimal point w.r.t to the clean data. (B) The gradient gives us the direction to move the suboptimal model to approach the true optimum. \label{update-arch2}}\vspace{-1em}
\end{figure}

The optimal clean model $\theta^{(c)}$ is visualized as a yellow star.
The first question is which direction to update $\theta^{(d)}$ (i.e., left or right).
For this class of models, given a suboptimal point, the direction to 
the global optimum is the gradient of the loss function.
The gradient is a $d$-dimensional vector function of the current model $\theta^{(d)}$ and the clean data.
Therefore, \sys needs to update $\theta^{(d)}$ some distance $\gamma$ (Figure \ref{update-arch2}B):
\[
\theta^{new} \leftarrow \theta^{(d)} - \gamma \cdot \nabla\phi(\theta^{(d)})
\]
At the optimal point, the magnitude of the gradient will be zero.
So intuitively, this approach iteratively moves the model downhill (transparent red circle) -- correcting the dirty model until the desired accuracy is reached.
However, the gradient depends on all of the clean data which is not available and \sys will have to approximate the gradient from a sample of newly cleaned data.
The main intuition is that if the gradient steps are on average correct, the model still moves downhill albeit with a reduced convergence rate proportional to the inaccuracy of the sample-based estimate.

To derive a sample-based update rule, the most important property is that sums commute with derivatives and gradients.
The convex loss class of models are sums of losses, so given the current best model $\theta$, the true gradient $g^*(\theta)$ is:
\[
g^*(\theta) = \nabla\phi(\theta) = \frac{1}{N} \sum_i^N \nabla\phi(x_i^{(c)},y_i^{(c)},\theta)
\]
\sys needs to estimate $g^*(\theta)$ from a sample $S$, which is drawn from the dirty data $R_{dirty}$.
 Therefore, the sum has two components the gradient on the already clean data $g_C$ which can be computed without cleaning and $g_S$ the gradient estimate from a sample of dirty data to be cleaned:
\[
g(\theta) = \frac{\mid R_{clean} \mid}{\mid R \mid} \cdot g_C(\theta) + \frac{\mid R_{dirty} \mid}{\mid R \mid} \cdot g_S(\theta)
\]
$g_C$ can be calculated by applying the gradient to all of the already cleaned records:
\[
g_C(\theta) = \frac{1}{\mid R_{clean}\mid}\sum_{i \in R_{clean}}\nabla\phi(x_i^{(c)},y_i^{(c)},\theta)
\]
$g_S$ can be estimated from a sample by taking the gradient w.r.t each record, and re-weighting the average by their respective sampling probabilities.
Before taking the gradient the cleaning function $C(\cdot)$ is applied to each sampled record.
Therefore, let $S$ be a sample of data, where each $i \in S$ is drawn with probability $p(i)$:
\[
g_{S}(\theta) = \frac{1}{\mid S \mid} \sum_{i \in S}\frac{1}{p(i)}\nabla\phi(x_i^{(c)},y_i^{(c)},\theta)
\]
Then, at each iteration $t$, the update becomes:
\[
\theta^{(t+1)} \leftarrow \theta^{(t)} - \gamma \cdot g(\theta^{(t)})
\]

\subsection{Model Update Algorithm}
To summarize, the algorithm is initialized with $\theta^{(0)} = \theta^{(d)}$ which is the dirty model.
There are three user set parameters the budget $k$, batch size $b$, and the step size $\gamma$.
In the following section, we will provide references from the convex optimization literature that allow the user to appropriately select these values.
At each iteration $t=\{1,...,T\}$, the cleaning is applied to a batch of data $b$ selected from the set of candidate dirty records $R_{dirty}$.
Then, an average gradient is estimated from the cleaned batch and the model is updated.
Iterations continue until $k = T \cdot b$ records are cleaned.

\begin{enumerate}[noitemsep]
	\item Calculate the gradient over the sample of newly clean data and call the result $g_S(\theta^{(t)})$
	\item Calculate the average gradient over all of the already clean data in $R_{clean}=R-R_{dirty}$, and call the result $g_C(\theta^{(t)})$
	\item Apply the following update rule:
	\[
	\theta^{(t+1)} \leftarrow \theta^{(t)} - \gamma \cdot(\frac{\mid R_{dirty} \mid}{\mid R \mid} \cdot g_S(\theta^{(t)}) + \frac{\mid R_{clean} \mid}{\mid R \mid} \cdot  g_C(\theta^{(t)}))
	\]
\end{enumerate} 

\subsection{Analysis with Stochastic Gradient Descent}\label{sgd}
The update algorithm can be formalized as a class of very well studied algorithms called Stochastic Gradient Descent.
SGD provides a theoretical framework to understand and analyze the update rule and bound the error.
Mini-batch stochastic gradient descent (SGD) is an algorithm for finding the optimal value
given the convex loss and data.
In mini-batch SGD, random subsets of data are selected at each iteration and the average gradient is computed for every batch.

One key difference with traditional SGD models is that \sys applies a \emph{full} gradient step on the already clean data and averages it with a stochastic gradient step (i.e., calculated from a sample) on the dirty data. 
Therefore, \sys iterations can take multiple passes over the clean data but at most a single cleaning pass of the dirty data.
The update algorithm can be thought of as a variant of SGD that lazily materializes the clean value.
As data is sampled at each iteration, data is cleaned when needed by the optimization.
It is well known that even for an arbitrary initialization SGD makes significant progress in less than one epoch (a pass through the entire dataset) \cite{bottou2012stochastic}.
In practice, the dirty model can be much more accurate than an arbitrary initialization as corruption may only affect a few features and combined with the full gradient step on the clean data the updates converge very quickly.

\vspace{0.25em}

\noindent\textbf{ Setting the step size $\gamma$: } There is extensive literature in machine learning for choosing the step size $\gamma$ appropriately. $\gamma$ can be set either to be a constant or decayed over time. Many machine learning frameworks (e.g., MLLib, Sci-kit Learn, Vowpal Wabbit) automatically set learning rates or provide different learning scheduling frameworks. 
In the experiments, we use a technique called inverse scaling where there is a parameter $\gamma_0=0.1$, and at each iteration it decays to $\gamma_t = \frac{\gamma_0}{\mid S \mid t}$. 

\vspace{0.25em}

\noindent\textbf{ Setting the batch size $b$: } The batch size should be set by the user to have the desired properties.
Larger batches will take longer to clean and will make more progress towards the clean model but will have less frequent model updates.
On the other hand, smaller batches are cleaned faster and have more frequent model updates.
There are diminishing returns to increasing the batch size $O(\frac{1}{\sqrt{b}})$.
In the experiments, we use a batch size of 50 which converges fast but allows for frequent model updates.
If a data cleaning technique requires a larger batch size than 50, i.e., data cleaning is fast enough that the iteration overhead is significant compared to cleaning 50 records, \sys can apply the updates in smaller batches.
For example, the batch size set by the user might be $b=1000$, but the model updates after every $50$ records are cleaned.
We can disassociate the batching requirements of SGD and the batching requirements of the data cleaning technique.

\subsubsection{Convergence Conditions and Properties}
Convergence properties of batch SGD formulations have been well studied \cite{dekel2012optimal}. Essentially, if the gradient estimate is unbiased and the step size is appropriately chosen, the algorithm is guaranteed to converge. 
In Appendix \ref{appsgd}, we show that the gradient estimate from \sys is indeed unbiased and our choice of step size is one that is established to converge.
The convergence rates of SGD are also well analyzed \cite{dekel2012optimal,bertsekas2011incremental,zhao2014stochastic}. 
The analysis gives a bound on the error of intermediate models and the expected number of steps before achieving a model within a certain error. 
For a general convex loss, a batch size $b$, and $T$ iterations, the convergence rate is bounded by $O(\frac{\sigma^2}{\sqrt{bT}})$. 
$\sigma^2$ is the variance in the estimate of the gradient at each iteration:
\[
\mathbb{E}(\|g - g^*\|^2)
\]
where $g^*$ is the gradient computed over the full data if it were fully cleaned.
This property of SGD allows us to bound the model error with a monotonically decreasing function of the number of records cleaned, thus satisfying the reliability condition in the problem statement.
If the loss in non-convex, the update procedure will converge towards a local minimum rather than the global minimum (See Appendix \ref{non-convex}).

\subsection{Example}
This example describes an application of the update algorithm.
\begin{example}\label{upex}
Recall that the analyst has a dirty SVM model on the dirty data $\theta^{(d)}$.
She decides that she has a budget of cleaning $100$ records, and decides to clean the 100 records in batches of 10 (set based on how fast she can clean the data, and how often she wants to see an updated result).
All of the data is initially treated as dirty with $R_{dirty} = R$ and $R_{clean} = \emptyset$.
The gradient of a basic SVM is given by the following function:
\[
\nabla\phi(x,y,\theta) =
\begin{cases}      
-y\cdot\boldsymbol{x} \text{ if } y\cdot\boldsymbol{x}\cdot\theta \le 1 \\
~~~~~~~0\ \text{ if } y\ \boldsymbol{x}\cdot\theta \geq 1      
\end{cases}
\]

For each iteration $t$, a sample of 10 records $S$ is drawn from $R_{dirty}$.
\sys then applies the cleaning function to the sample:
\[
\{(x_i^{(c)},y_i^{(c)})\} = \{C(i): \forall i \in S\}
\]
Using these values, \sys estimates the gradient on the newly cleaned data:
\[
g_{S}(\theta) = \frac{1}{10} \sum_{i \in S}\frac{1}{p(i)}\nabla\phi(x_i^{(c)},y_i^{(c)},\theta)
\]
\sys also applies the gradient to the already clean data (initially non-existent):
\[
g_C(\theta) = \frac{1}{\mid R_{clean}\mid}\sum_{i \in R_{clean}}\nabla\phi(x_i^{(c)},y_i^{(c)},\theta)
\]
Then, it calculates the update rule:
\[
	\theta^{(t+1)} \leftarrow \theta^{(t)} - \gamma \cdot(\frac{\mid R_{dirty} \mid}{\mid R \mid} \cdot g_S(\theta^{(t)}) + \frac{\mid R_{clean} \mid}{\mid R \mid} \cdot  g_C(\theta^{(t)}))
\] 
Finally, $R_{dirty} \leftarrow R_{dirty} - S$, $R_{clean} \leftarrow R_{clean} + S$, and continue to the next iteration.
\end{example}

%% file: sampling.tex
\section{Efficiency With Sampling}\label{dist-samp}
The \emph{updater} received a sample with probabilities $p(\cdot)$.
For any distribution where  $p(\cdot) > 0$, we can preserve correctness.
\sys uses a sampling algorithm that selects the most valuable records to clean with higher probability. 

\subsection{Oracle Sampling Problem}
Recall that the convergence rate of an SGD algorithm is bounded by $\sigma^2$ which is the variance of the gradient.
Intuitively, the variance measures how accurately the gradient is estimated from a uniform sample.
Other sampling distributions, while preserving the sample expected value, may have a lower variance.
Thus, the oracle sampling problem is defined as a search over sampling distributions to find the minimum variance sampling distribution.

\begin{definition}[Oracle Sampling Problem]
Given a set of candidate dirty data $R_{dirty}$, $\forall r \in R_{dirty}$ find sampling probabilities $p(r)$ such that over all samples $S$ of size $k$ it minimizes:
\[
\mathbb{E}(\|g_S - g^*\|^2)
\]
\end{definition}
It can be shown that the optimal distribution over records in $R_{dirty}$ is probabilities proportional to:
\[
p_i \propto \|\nabla\phi(x^{(c)}_i,y^{(c)}_i,\theta^{(t)})\|
\]
This is an established result, for thoroughness, we provide a proof in the appendix (Section \ref{impsample-deriv}), but intuitively, records with higher gradients should be sampled with higher probability as they affect the update more significantly.
However, \sys cannot exclude records with lower gradients as that would induce a bias hurting convergence.
The problem is that the optimal distribution leads to a chicken-and-egg problem:
the optimal sampling distribution requires knowing $(x^{(c)}_i,y^{(c)}_i)$, however, cleaning is required to know those values.

\subsection{Dirty Gradient Solution}\label{dgsample}
Such an oracle does not exist, and one solution is to use the gradient w.r.t to the dirty data:
\[
p_i \propto \|\nabla\phi(x^{(d)}_i,y^{(d)}_i,\theta^{(t)})\|
\]
It turns out that the solution works reasonably well in practice on our experimental datasets and has been studied in Machine Learning as the Expected Gradient Length heuristic \cite{settles2010active}.
The contribution in this work is integrating this heuristic with statistically correct updates.
However, intuitively, approximating the oracle as closely as possible can result in improved prioritization.
The subsequent section describes two components, the detector and estimator, that can be used to improve the convergence rate.
Our experiments suggest up-to a 2x improvement in convergence when using these optional optimizations (Section \ref{comp}).

%% file: optimal.tex
\section{Optimizations}\label{opti}

In this section, we describe two approaches to optimization, the {\it Detector} and the {\it Estimator}, that
improve the efficiency of the cleaning process.  
Both approaches are designed to increase the likelihood that the 
{\it Sampler} will pick dirty records that, once cleaned,
most move the model towards the true clean model.
The {\it Detector} is intended to learn the characteristics that distinguish dirty records from clean records
while the {\it Estimator} is designed to estimate the amount that cleaning a given dirty record will move the 
model towards the true optimal model.

\input{detector}

\input{estimator}

%% file: detector.tex
\subsection{The Detector}\label{det}

The detector returns two important aspects of a record: 
(1) whether the record is dirty, and (2) if it is dirty, what is wrong with the record.
The sampler can use (1) to select a subset of dirty records to sample at each batch and 
the estimator can use (2) estimate the value of data cleaning based on other records with the same corruption.
\sys supports two types of detectors: \emph{a priori} and \emph{adaptive}.
In former assumes that we know the set of dirty records and how they are dirty \emph{a priori} to \sys,
while the latter \emph{adaptively} learns characteristics of the dirty data as part of running \sys.

\subsubsection{\protect\textit{\large A Priori} Detector}
For many types of dirtiness such as missing attribute values and constraint violations, 
it is possible to efficiently enumerate a set of corrupted records and determine how the records are corrupted.

\begin{definition}[A Priori Detection]
Let $r$ be a record in $R$. An a priori detector is a detector that returns a Boolean of whether the record is dirty and a set of columns $e_r$ that are dirty.
\[
D(r) = (\{0,1\}, e_r)
\]
From the set of columns that are dirty, find the corresponding features that are dirty $f_r$ and labels that are dirty $l_r$.
\end{definition}

\noindent Here is an example this definition using a data cleaning methodology proposed in the literature.

\vspace{0.5em}

\noindent\textbf{Constraint-based Repair: }
One model for detecting errors involves declaring constraints on the database.

\vspace{0.5em}

\emph{Detection. } Let $\Sigma$ be a set of constraints on the relation $\mathcal{R}$. 
In the detection step, the detector selects a subset of records $\mathcal{R}_{dirty} \subseteq \mathcal{R}$ that violate at least one constraint.
The set $e_r$ is the set of columns for each record which have a constraint violation. 

\begin{example}\label{detex1}
An example of a constraint on the running example dataset is that the \texttt{status} of
a contribution can be only ``allowed" or ``disallowed".
Any other value for \texttt{status} is an error.
\end{example}

\subsection{Adaptive Detection}
\emph{A priori} detection is not possible in all cases.
The detector also supports adaptive detection where detection is learned from previously cleaned data.
Note that this ``learning" is distinct from the ``learning" at the end of the pipeline.
The challenge in formulating this problem is that detector needs to describe how the data is dirty (e.g. $e_r$ in the \emph{a priori} case).
The detector achieves this by categorizing the corruption into $u$ classes.
These classes are corruption categories that do not necessarily align with features, but every record is classified with at most one category.

When using adaptive detection, the repair step has to clean the data and report to which of the $u$ classes the corrupted record belongs.
When an example $(x,y)$ is cleaned, the repair step labels it with one of the ${\text{clean}, 1,2,...,u+1}$ classes (including one for ``not dirty").
It is possible that $u$ increases each iteration as more types of dirtiness are discovered.
In many real world datasets, data errors have locality, where similar records tend to be similarly corrupted.
There are usually a small number of error classes even if a large number of records are corrupted.

One approach for adaptive detection is using a statistical classifier. 
This approach is particularly suited for a small number data error classes, each of which containing many erroneous records.
This problem can be addressed by any classifier, and we use an all-versus-one SVM in our experiments.

Another approach could be to adaptively learn predicates that define each of the error classes.
For example, if records with certain attributes are corrupted, a pattern tableau can be assigned to each class to select a set of possibly corrupted records.
This approach is better suited than a statistical approach for a large number of error classes or scarcity of errors.
However, it relies on errors being well aligned with certain attribute values.

\begin{definition}[Adaptive Case]
Select the set of records for which $\kappa$ gives a positive error classification (i.e., one of the $u$ error classes).
After each sample of data is cleaned, the classifier $\kappa$ is retrained.
So the result is:
\[D(r) = (\{1,0\},\{1,...,u+1\})\]
\end{definition}

\vspace{0.75em}

\noindent\textbf{Adaptive Detection With OpenRefine: }
\begin{example}\label{detex2}
OpenRefine is a spreadsheet-based tool that allows users to explore and transform data.
However, it is limited to cleaning data that can fit in memory on a single computer.
Since the cleaning operations are coupled with data exploration, \sys does not know what is dirty in advance (the analyst may discover new errors as she cleans).

Suppose the analyst wants to use OpenRefine to clean the running example dataset with \sys.
She takes a sample of data from the entire dataset and uses the tool to discover errors.
For example, she finds that some drugs are incorrectly classified as both drugs and devices.
She then removes the device attribute for all records that have the drug name in question.
As she fixes the records, she tags each one with a category tag of which corruption it belongs to.
\end{example}

%% file: estimator.tex
\subsection{The Estimator}\label{sampling}
To get around the problem with oracle sampling, the estimator will estimate the cleaned value with previously cleaned data.
The estimator will also take advantage of the detector from the previous section.
There are a number of different approaches, such as regression, that could be used to estimate the cleaned value given the dirty values.
However, there is a problem of scarcity, where errors may affect a small number of records.
As a result, the regression approach would have to learn a multivariate function with only a few examples.
Thus, high-dimensional regression ill-suited for the estimator.
Conversely, it could try a very simple estimator that just calculates an average change and adds this change to all of the gradients.
This estimator can be highly inaccurate as it also applies the change to records that are known to be clean.

\sys leverages the detector for an estimator between these two extremes.
The estimator calculates average changes feature-by-feature and selectively corrects the gradient when a feature is known to be corrupted based on the detector.
It also applies a linearization that leads to improved estimates when the sample size is small.
We evaluate the linearization in Section \ref{est} against alternatives, and find that it provides more accurate estimates for a small number of samples cleaned.
The result is a biased estimator, and when the number of cleaned samples is large the alternative techniques are comparable or even slightly better due to the bias.

\paragraph{Estimation For A Priori Detection}
If most of the features are correct, it would seem like the gradient is only
incorrect in one or two of its components.
The problem is that the gradient $\nabla\phi(\cdot)$ can be a very non-linear function of the features that couple features together.
For example, the gradient for linear regression is:
\[
\nabla\phi(x,y,\theta) = (\theta^Tx - y)x
\]
It is not possible to isolate the effect of a change of one feature on the gradient.
Even if one of the features is corrupted, all of the gradient components will be incorrect.

To address this problem, the gradient can be approximated in a way that the effects of dirty features on the gradient are decoupled.
Recall, in the \emph{a priori} detection problem, that associated with each $r \in R_{dirty}$ is a set of errors $f_r,l_r$ which is a set that identifies a set of corrupted features and labels.
This property can be used to construct a coarse estimate of the clean value.
The main idea is to calculate average changes for each feature, then given an uncleaned (but dirty) record, add these average changes to correct the gradient.

To formalize the intuition, instead of computing the actual gradient with respect to the 
true clean values, compute the conditional expectation given that a set of features and labels $f_r,l_r$ are corrupted:
\[
p_i \propto \mathbb{E}(\nabla\phi(x^{(c)}_i,y^{(c)}_i,\theta^{(t)}) \mid f_r,l_r)
\]
Corrupted features are defined as that:
\[
i \notin f_r \implies x^{(c)}[i] - x^{(d)}[i] = 0
\]
\[
i \notin l_r \implies y^{(c)}[i] - y^{(d)}[i] = 0
\]

The needed approximation represents a linearization of the errors, and the resulting approximation will be of the form:
\[
p(r)\propto\|\nabla\phi(x,y,\theta^{(t)}) + M_x \cdot \Delta_{rx} +  M_y \cdot \Delta_{ry}\|
\]
where $M_x$, $M_y$ are matrices and $\Delta_{rx}$ and $\Delta_{ry}$ are vectors with one component for each feature and label where each value is the average change for those features that are corrupted and 0 otherwise.
Essentially, it the gradient with respect to the dirty data plus some linear correction factor.
In the appendix, we present a derivation using a Taylor series expansion and a number of $M_x$ and $M_y$ matrices for common convex losses (Appendix \ref{apptaylor} and \ref{example-deriv}).
The appendix also describes how to maintain $\Delta_{rx}$ and $\Delta_{ry}$ as cleaning progresses.

\paragraph{Estimation For Adaptive Case}
A similar procedure holds in the adaptive setting, however, it requires reformulation.
Here, \sys uses $u$ corruption classes provided by the detector.
Instead of conditioning on the features that are corrupted, the estimator conditions on the classes.
So for each error class, it computes a $\Delta_{ux}$ and $\Delta_{uy}$.
These are the average change in the features given that class and the average change in labels given that class.
\[
p(r_u)\propto\|\nabla\phi(x,y,\theta^{(t)}) + M_x \cdot \Delta_{ux} +  M_y \cdot \Delta_{uy}\|
\] 

\subsubsection{Example}
Here is an example of using the optimization to select a sample of data for cleaning.
\begin{example}\label{estex}
Consider using \sys with an a priori detector.
Let us assume that there are no errors in the labels and only errors in the features.
Then, each training example will have a set of corrupted features (e.g., $\{1,2,6\}$, $\{1,2,15\}$).
Suppose that the cleaner has just cleaned the records $r_1$ and $r_2$ represented as tuples with their corrupted feature set: ($r_1$,$\{1,2,3\}$), ($r_2$,$\{1,2,6\}$).
For each feature $i$, \sys maintains the average change between dirty and clean in a value in a vector $\Delta_x[i]$ for those records corrupted on that feature. 

Then, given a new record ($r_3$,$\{1,2,3,6\}$), $\Delta_{r_3x}$ is the vector $\Delta_x$ where component $i$ is set to 0 if the feature is not corrupted.
Suppose the data analyst is using an SVM, then the $M_x$ matrix is as follows:
\[
M_x[i,i] = \begin{cases}      
-y[i] \text{ if } y\cdot\boldsymbol{x}\cdot\theta \le 1 \\
0\ \text{ if } y\ \boldsymbol{x}\cdot\theta \geq 1      
\end{cases} 
\]
Thus, we calculate a sampling weight for record $r_3$:
\[
p(r_3) \propto\|\nabla\phi(x,y,\theta^{(t)}) + M_x \cdot \Delta_{r_3x} \|
\] 
To turn the result into a probability distribution, \sys normalizes over all dirty records.
\end{example}

%% file: experiments.tex
\section{Experiments}\label{eval}
First, the experiments evaluate how various types of corrupted data benefit from data cleaning.
Next, the experiments explore different prioritization and model update schemes for progressive data cleaning.
Finally, \sys is evaluated end-to-end in a number of real-world data cleaning scenarios.

\subsection{Experimental Setup and Notation}
The main metric for evaluation is a relative measure of the trained model and the model if all of the data is cleaned.

\vspace{0.5em}

\noindent\textbf{Relative Model Error. } Let $\theta$ be the model trained on the dirty data, and let $\theta^*$ be the model trained on the same data if it was cleaned. Then the model error is defined as $\frac{\|\theta - \theta^*\|}{\|\theta^*\|}$.

\subsubsection{Scenarios}


\vspace{0.25em}

\noindent\textbf{Income Classification (Adult): } In this dataset of 45,552 records, the task is to predict the income bracket (binary) from 12 numerical and categorical covariates with an SVM classifier. 

\vspace{0.25em}

\noindent\textbf{Seizure Classification (EEG): } In this dataset, the task is to predict the onset of a seizure (binary) from 15 numerical covariates with a thresholded Linear Regression. There are 14980 data points in this dataset. This classification task is inherently hard with an accuracy on completely clean data of only 65\%.

\vspace{0.25em}

\noindent\textbf{Handwriting Recognition (MNIST) \footnote{\scriptsize\url{http://ufldl.stanford.edu/wiki/index.php/Using_the_MNIST_Dataset}}: } In this dataset, the task is to classify 60,000 images of handwritten images into 10 categories with an one-to-all multiclass SVM classifier. The unique part of this dataset is the featurized data consists of a 784 dimensional vector which includes edge detectors and raw image patches. 

\vspace{0.25em}

\noindent\textbf{Dollars For Docs: } The dataset has 240,089 records with 5 textual attributes and one numerical attribute.
The dataset is featurized with bag-of-words featurization model for the textual attributes which resulted in a 2021 dimensional feature vector, and a binary SVM is used to classify the status of the medical donations.

\vspace{0.25em}

\noindent\textbf{World Bank: } The dataset has 193 records of country name, population, and various macro-economics statistics. The values are listed with the date at which they were acquired. This allowed us to determine that records from smaller and less populous countries were more likely to be out-of-date.

\subsubsection{Compared Algorithms}
\noindent Here are the alternative methodologies evaluated in the experiments:

\vspace{0.25em}

\noindent\textbf{Robust Logistic Regression \cite{feng2014robust}. } Feng et al. proposed a variant of logistic regression that is robust to outliers. We chose this algorithm because it is a robust extension of the convex regularized loss model, leading to a better apples-to-apples comparison between the techniques. (See details in Appendix \ref{rlogit})  

\vspace{0.25em}

\noindent\textbf{Discarding Dirty Data. } As a baseline, dirty data are discarded.

\vspace{0.25em}

\noindent\textbf{SampleClean (SC) \cite{wang1999sample}. } SampleClean takes a sample of data, applies data cleaning, and then trains a model to completion on the sample.

\vspace{0.25em}

\noindent\textbf{Active Learning (AL) \cite{guillory2009active}. } To fairly evaluate Active Learning, we first apply our gradient update to ensure correctness.
Within each iteration, examples are prioritized by distance to the decision boundary (called Uncertainty Sampling in \cite{settles2010active}).
However, we do not include our optimizations such as detection and estimation.

\vspace{0.25em}

\noindent\textbf{ActiveClean Oracle (AC+O): } In \sys Oracle, instead of an estimation and detection step, the true clean value is used to evaluate the theoretical ideal performance of \sys.

\subsection{Does Data Cleaning Matter?}
The first experiment evaluates the benefits of data cleaning on two of the example datasets (EEG and Adult).
Our goal is to understand which types of data corruption are amenable to data cleaning and which are better suited for robust statistical techniques.
The experiment compares four schemes: (1) full data cleaning  , (2) baseline of no cleaning, (3) discarding the dirty data, and (4) robust logistic regression,. We corrupted 5\% of the training examples in each dataset in two different ways:

\vspace{0.5em}

\noindent\textbf{Random Corruption: } Simulated high-magnitude random outliers. 5\% of the examples are selected at random and a random feature is replaced with 3 times the highest feature value.

\vspace{0.5em}

\noindent\textbf{Systematic Corruption: } Simulated innocuous looking (but still incorrect) systematic corruption. The model is trained on the clean data, and the three most important features (highest weighted) are identified. The examples are sorted by each of these features and the top examples are corrupted with the mean value for that feature (5\% corruption in all). 
It is important to note that examples can have multiple corrupted features.

\begin{figure}[t]
\centering
 \includegraphics[width=0.49\columnwidth]{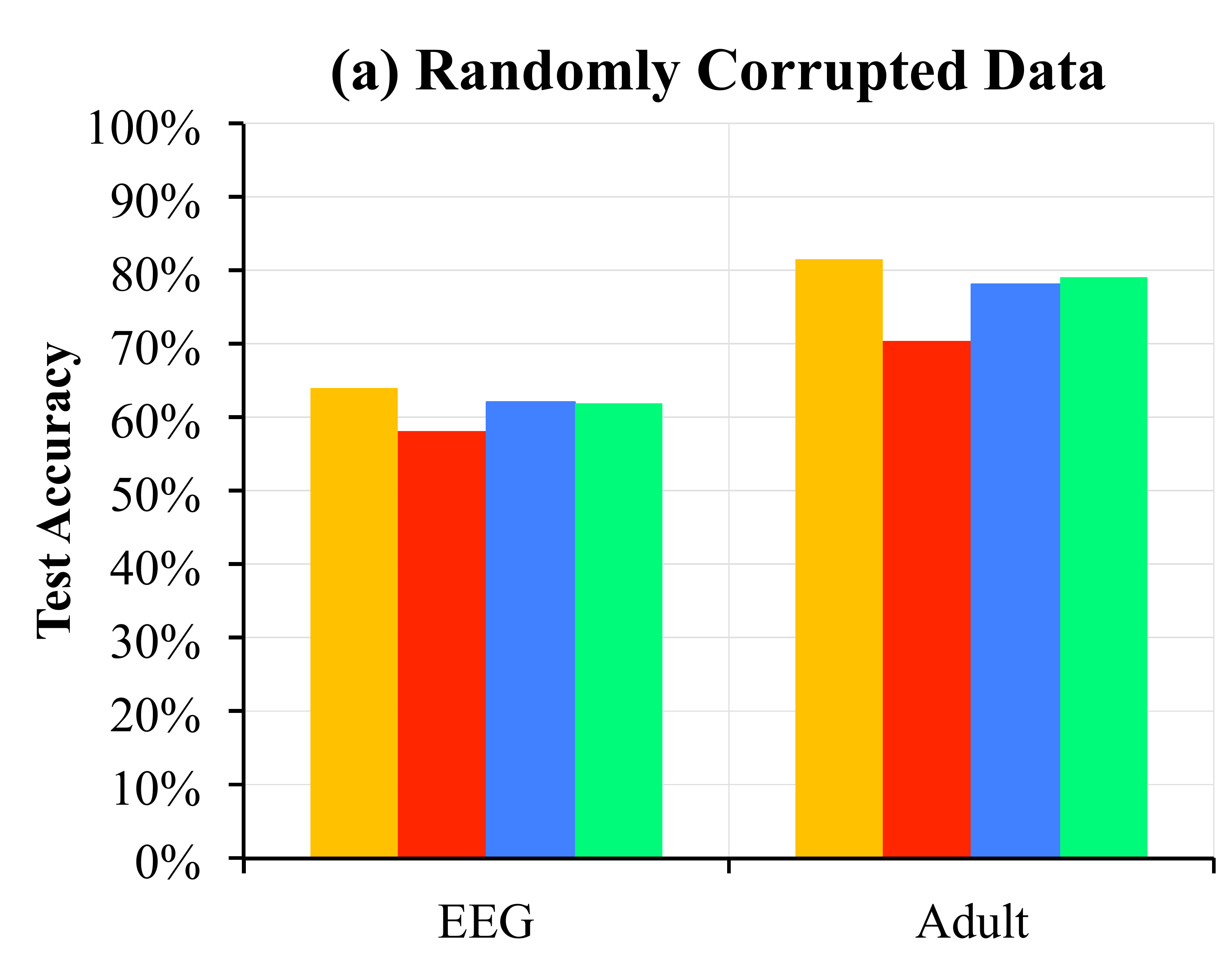}
 \includegraphics[width=0.49\columnwidth]{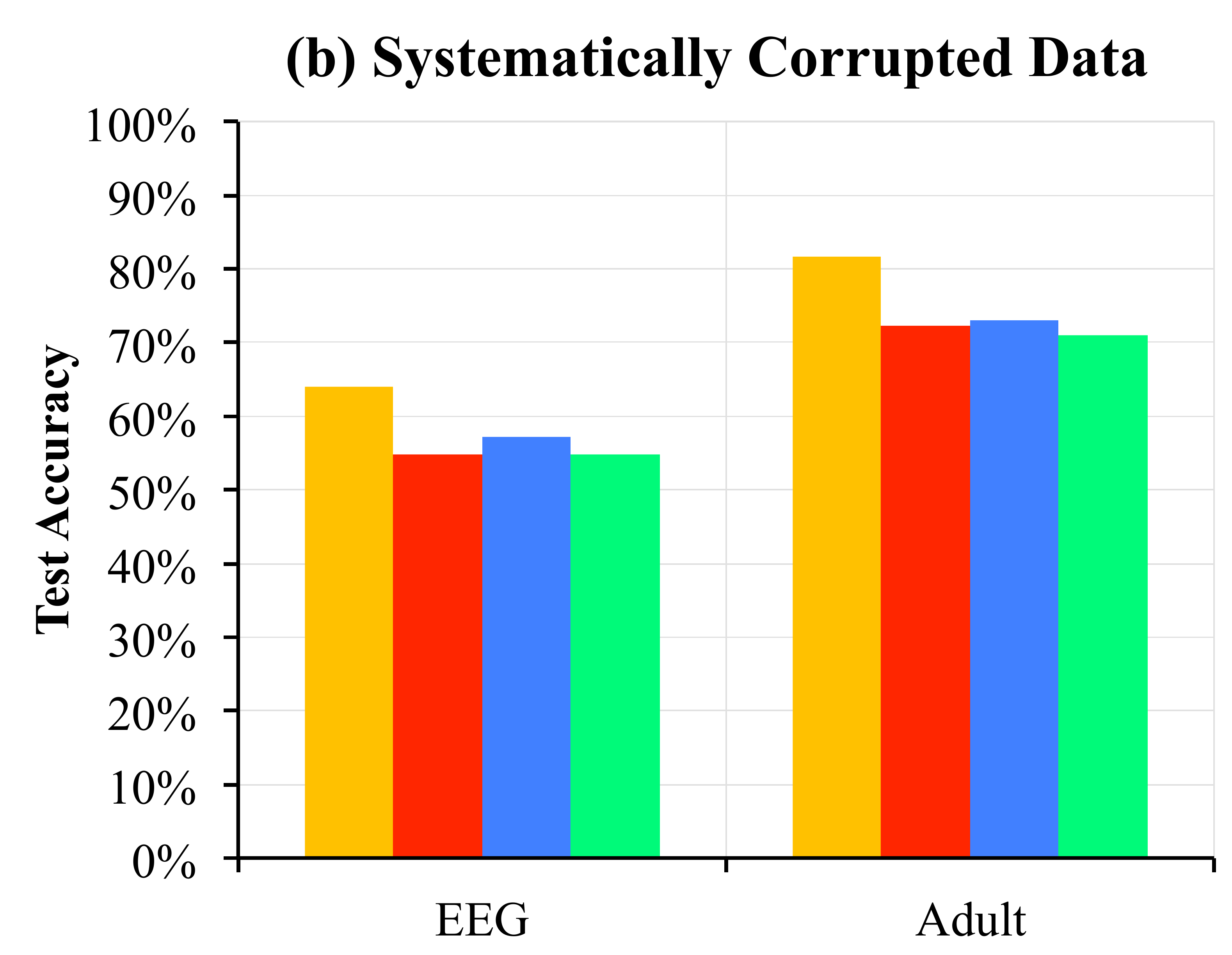}
 \includegraphics[width=0.5\columnwidth]{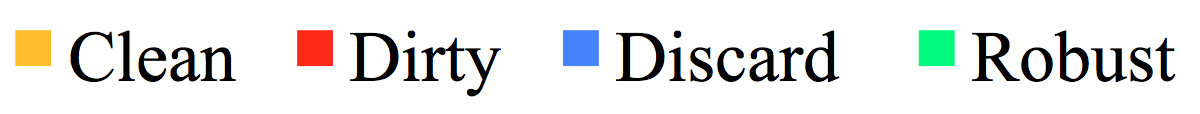}\vspace{-1em}
 \caption{(a) Robust techniques and discarding data work when corrupted data are random and look atypical. (b) Data cleaning can provide reliable performance in both the systematically corrupted setting and randomly corrupted setting.\label{sys-rand}}\vspace{-1.5em}
\end{figure}

Figure \ref{sys-rand} shows the test accuracy for models trained on both types of data with the different techniques.
The robust method performs well on the random high-magnitude outliers with only a 2.0\% reduction in clean test accuracy for EEG and 2.5\% reduction for Adult.
In the random setting, discarding dirty data also performs relatively well.
However, the robust method falters on the systematic corruption with a 9.1\% reduction in clean test accuracy for EEG and 10.5\% reduction for Adult.
The problem is that without cleaning, there is no way to know if the corruption is random or systematic and when to trust a robust method.
While data cleaning requires more effort, it provides benefits in both settings.
In the remaining experiments, unless otherwise noted, the experiments use systematic corruption.

\noindent \emph{Summary: A 5\% systematic corruption can introduce a 10\% reduction in test accuracy even when using a robust method.}

\subsection{\sys: \protect\textit{\large A Priori} Detection}
The next set of experiments evaluate different approaches to cleaning a sample of data compared to \sys using \emph{a priori} detection.
\emph{A priori} detection assumes that all of the corrupted records are known in advance but their clean values are unknown. 

\subsubsection{Active Learning and SampleClean}
The next experiment evaluates the samples-to-error tradeoff between four alternative algorithms: \sys (AC), SampleClean, Active Learning, and \sys+Oracle (AC+O).
Figure \ref{prio-perf} shows the model error and test accuracy as a function of the number of cleaned records.
In terms of model error, \sys gives its largest benefits for small sample sizes.
For 500 cleaned records of the Adult dataset, \sys has 6.1x less error than SampleClean and 2.1x less error than Active Learning.
For 500 cleaned records of the EEG dataset, \sys has 9.6x less error than SampleClean and 2.4x less error than Active Learning.
Both Active Learning and \sys benefit from the initialization with the dirty model as they do not retrain their models from scratch, and \sys improves on this performance with detection and error estimation.
Active Learning has no notion of dirty and clean data, and therefore prioritizes with respect to the dirty data.
These gains in model error also correlate well to improvements in test error (defined as the test accuracy difference w.r.t cleaning all data).
The test error converges more quickly than model error, emphasizing the benefits of progressive data cleaning, since it is not neccessary to clean all the data to get a model with essentially the same performance as the clean model.
For example, to achieve a test error of 1\% on the Adult dataset, \sys cleans 500 fewer records than Active Learning.

\vspace{0.25em}

\noindent \emph{Summary: \sys with a priori detection returns results that are more than 6x more accurate than SampleClean and 2x more accurate than Active Learning for cleaning 500 records.}

\begin{figure}[t]
\centering\vspace{-1em}
 \includegraphics[width=0.49\columnwidth]{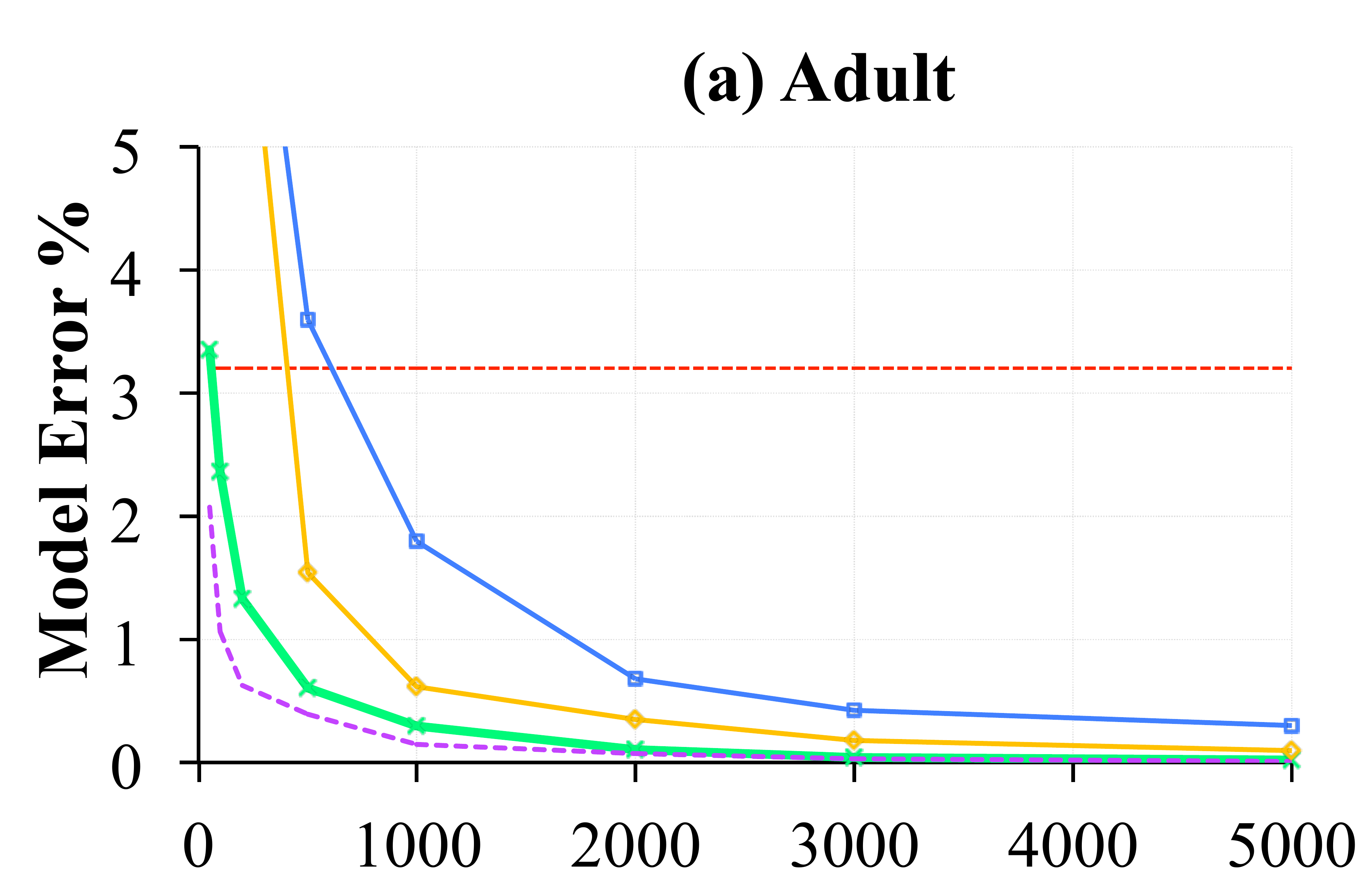}
  \includegraphics[width=0.49\columnwidth]{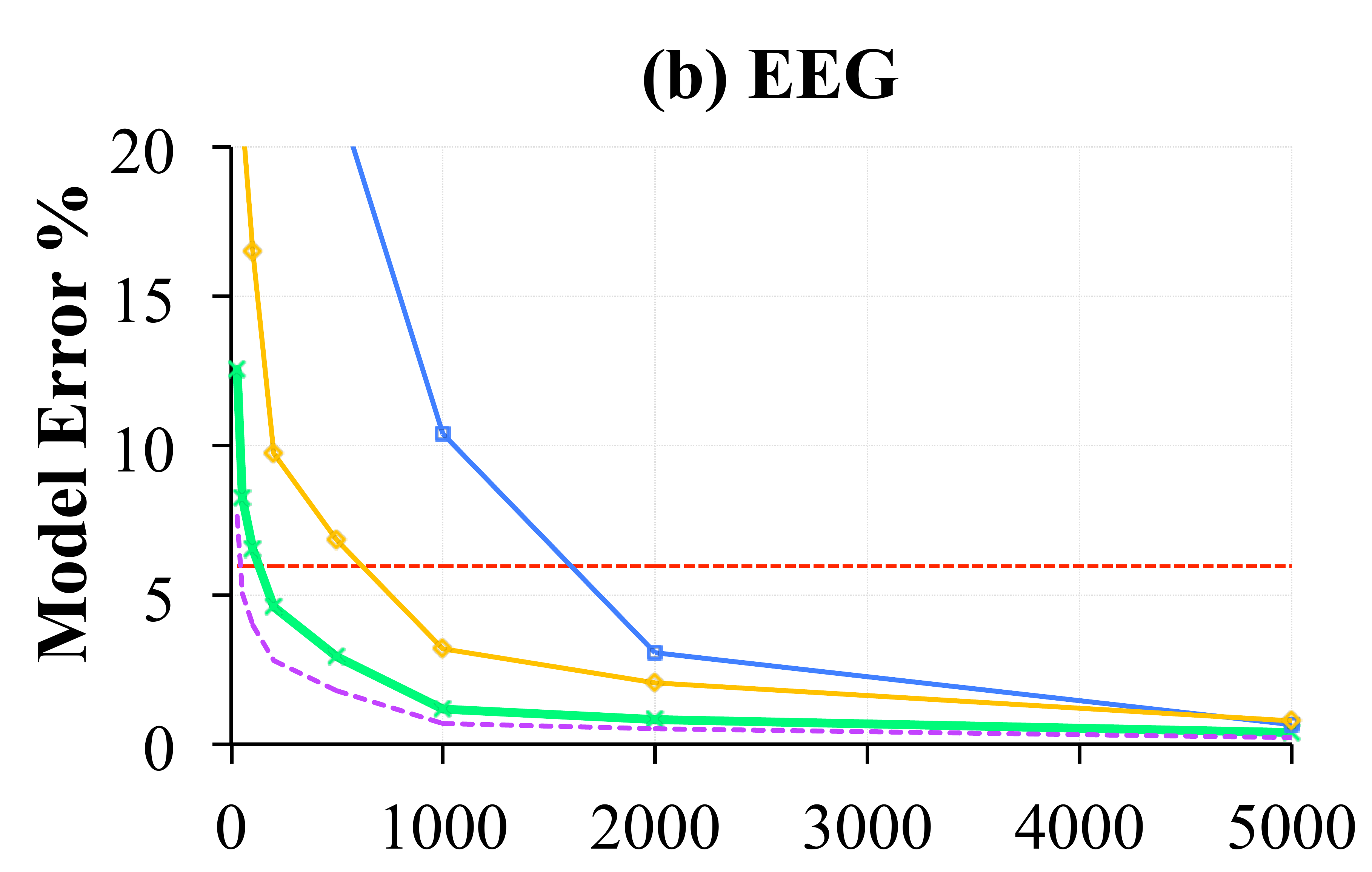}
  \includegraphics[width=0.49\columnwidth]{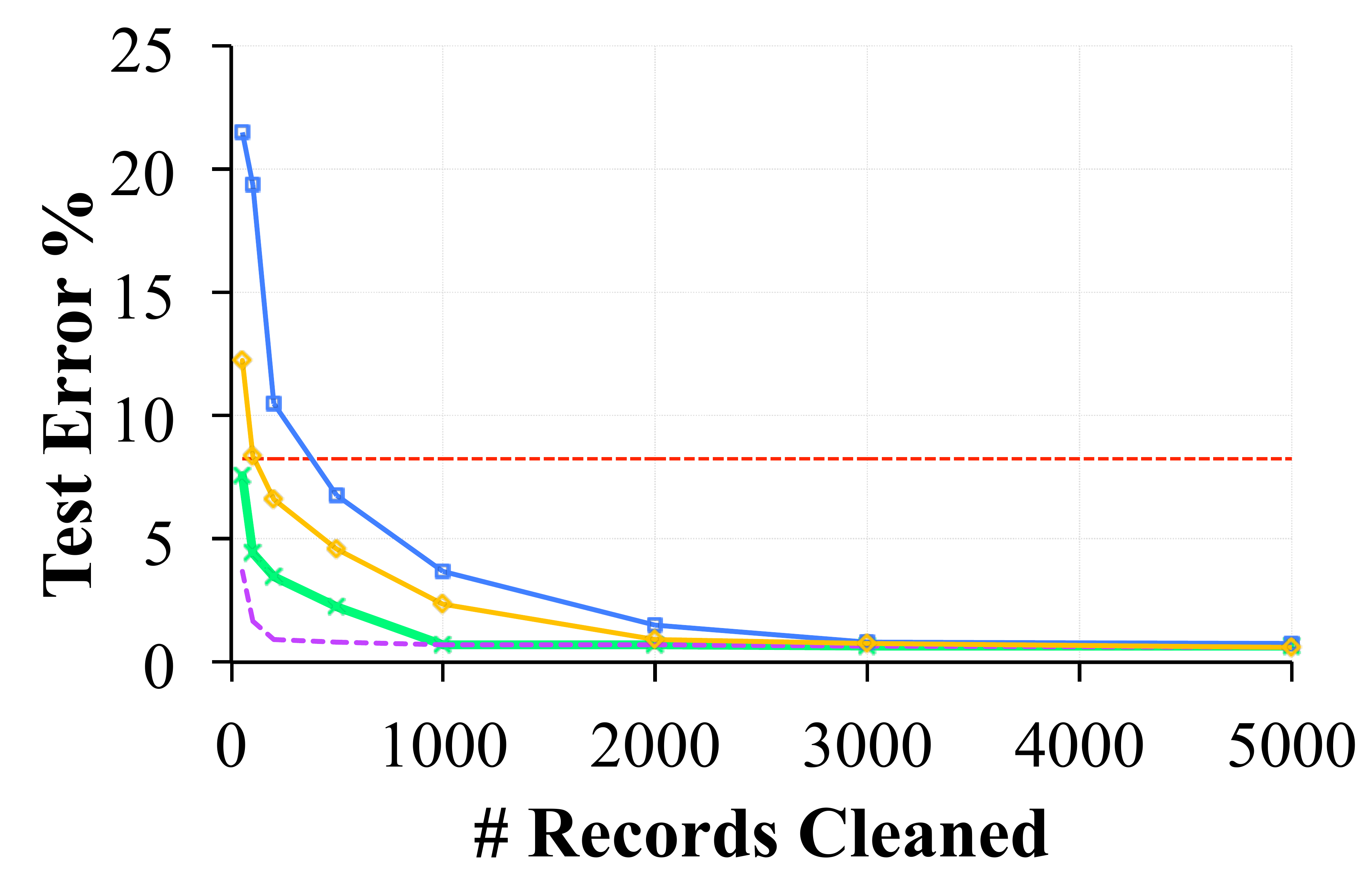}
  \includegraphics[width=0.49\columnwidth]{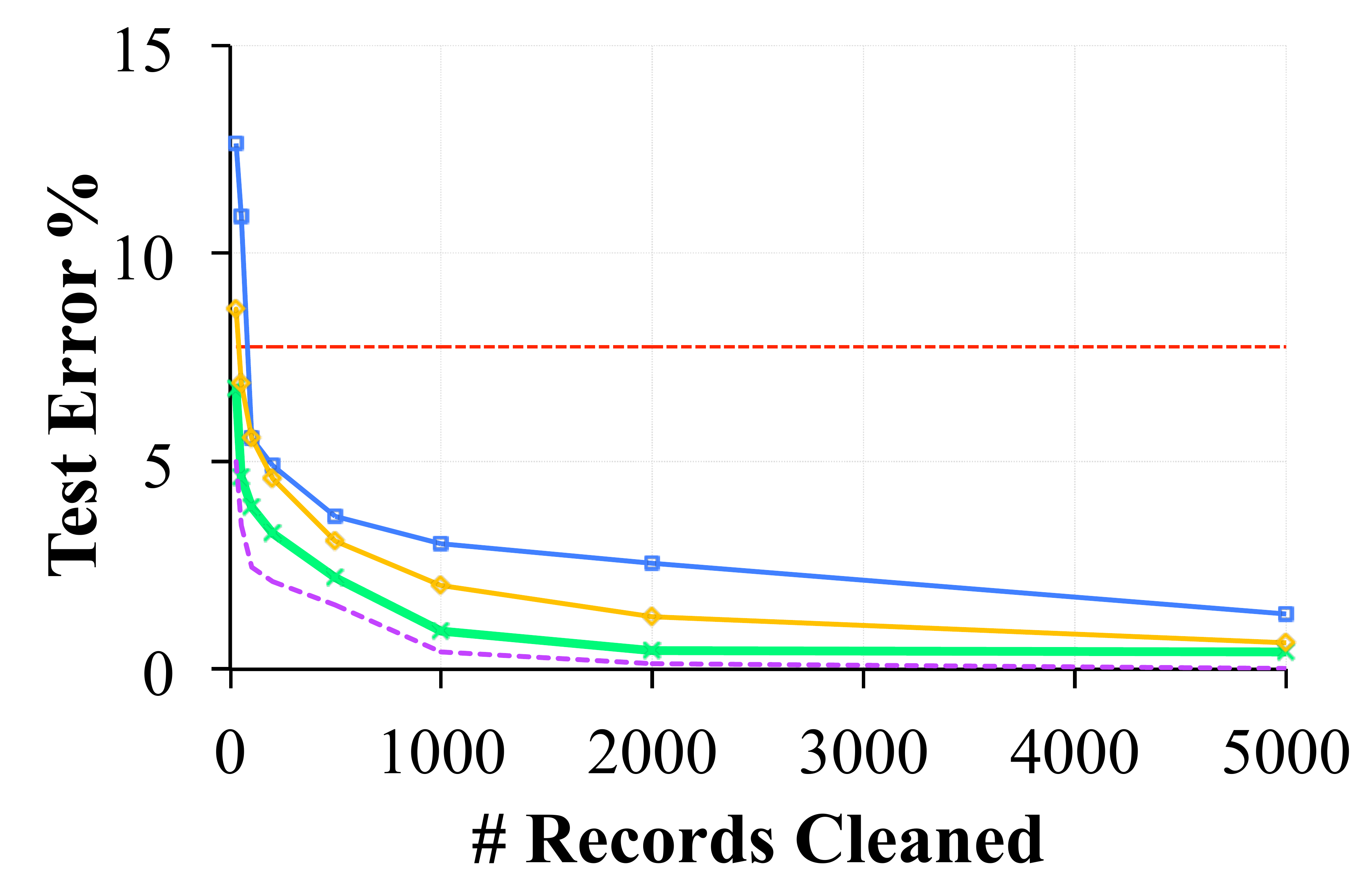}
  \includegraphics[width=0.5\columnwidth]{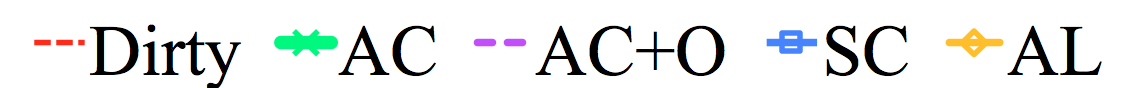}\vspace{-0.5em}
 \caption{ The relative model error as a function of the number of examples cleaned. \sys converges with a smaller sample size to the true result in comparison to Active Learning and SampleClean. \label{prio-perf}}\vspace{-1em}
\end{figure}

\subsubsection{Source of Improvements}\label{comp}
The next experiment compares the performance of \sys with and without various optimizations at 500 records cleaned point. 
\sys without detection is denoted as (AC-D) (that is at each iteration we sample from the entire dirty data), and \sys without detection and importance sampling is denoted as (AC-D-I).
Figure \ref{opts} plots the relative error of the alternatives and \sys with and without the optimizations.
Without detection (AC-D), \sys is still more accurate than Active Learning.
Removing the importance sampling, \sys is slightly worse than Active Learning on the Adult dataset but is comparable on the EEG dataset.

\begin{figure}[t]\vspace{0.5em}
\centering
 \includegraphics[width=0.49\columnwidth]{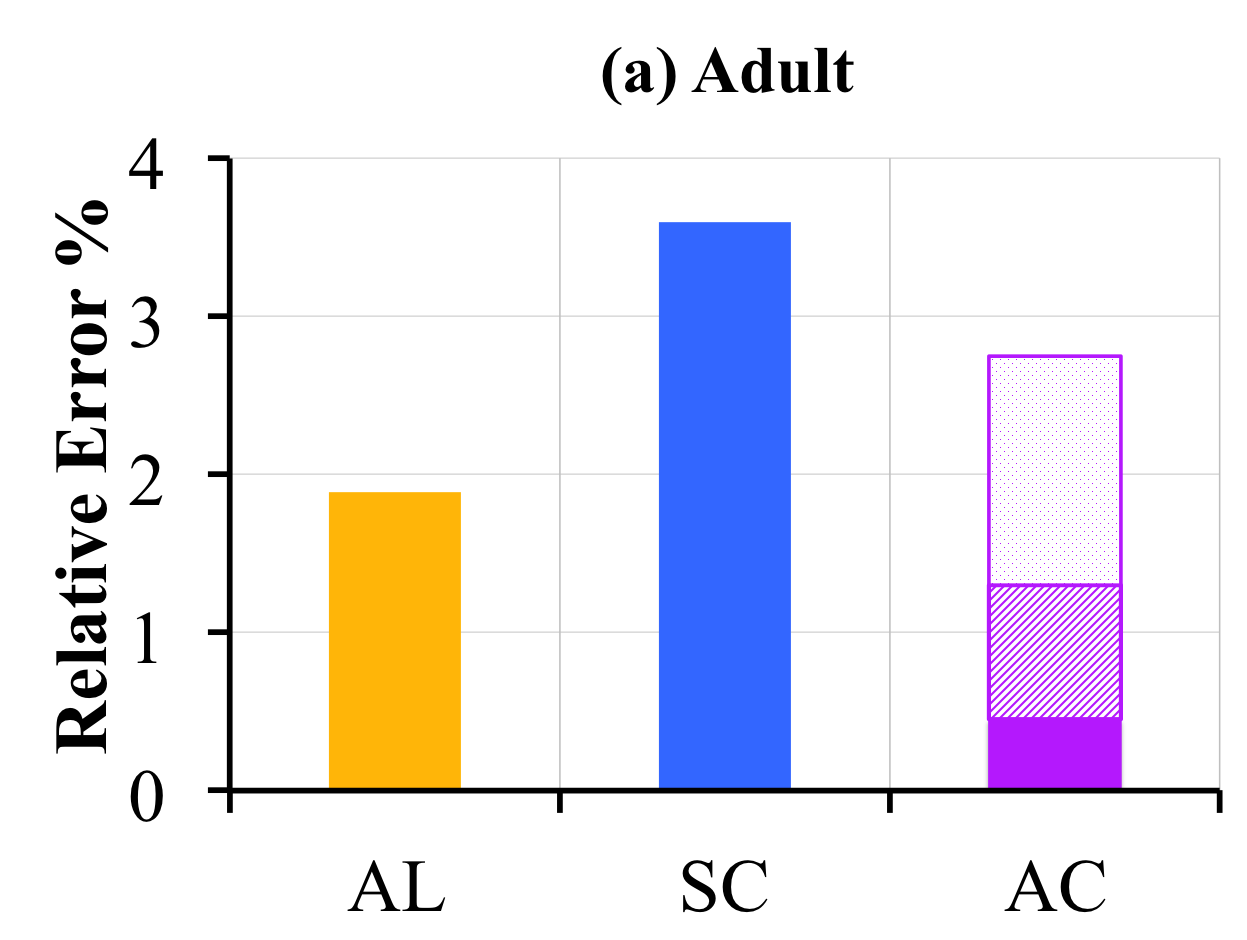}
 \includegraphics[width=0.49\columnwidth]{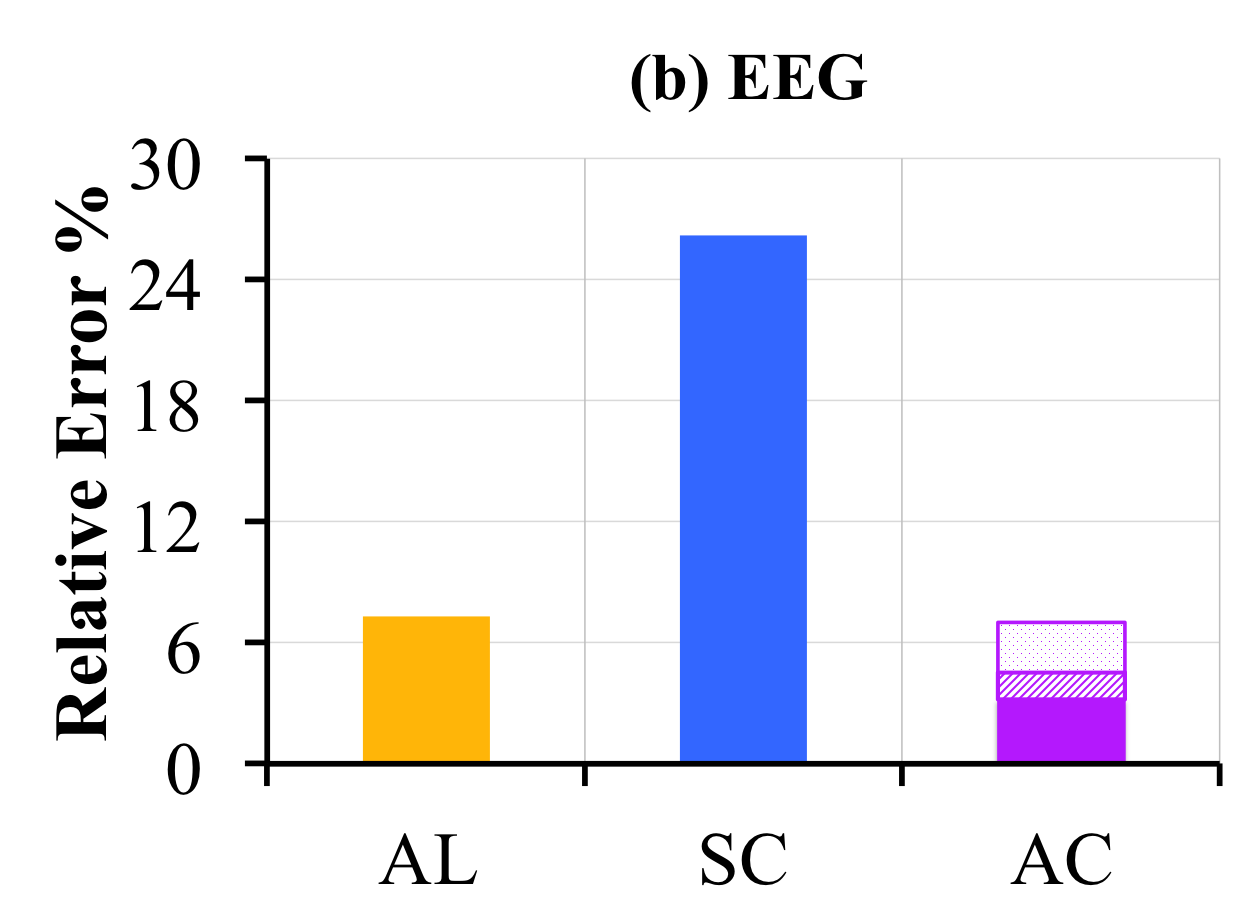}
 \includegraphics[width=0.5\columnwidth]{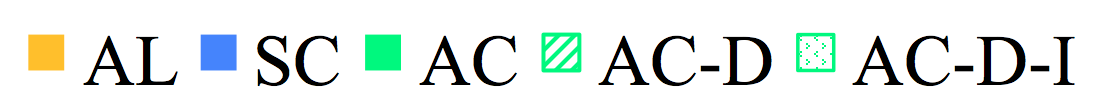}\vspace{-1em}
 \caption{ -D denotes no detection, and -D-I denotes no detection and no importance sampling. Both optimizations significantly help \sys outperform SampleClean and Active Learning. \label{opts}}\vspace{-1.5em}
\end{figure}

\vspace{0.25em}

\noindent \emph{Summary: Both a priori detection and non-uniform sampling significantly contribute to the gains over Active Learning.}

\subsubsection{Mixing Dirty and Clean Data}
Training a model on mixed data is an unreliable methodology lacking the same guarantees as Active Learning or SampleClean even in the simplest of cases.
For thoroughness, the next experiments include the model error as a function of records cleaned in comparison to \sys.
Figure \ref{pc-perf} plots the same curves as the previous experiment comparing \sys, Active Learning, and two mixed data algorithms.
PC randomly samples data, clean, and writes-back the cleaned data.
PC+D randomly samples data from using the dirty data detector, cleans, and writes-back the cleaned data.
For these errors PC and PC+D give reasonable results (not always guaranteed), but \sys converges faster.
\sys tunes the weighting when averaging dirty and clean data into the gradient.

\begin{figure}[ht!]
\centering\vspace{-0.5em}
 \includegraphics[width=0.49\columnwidth]{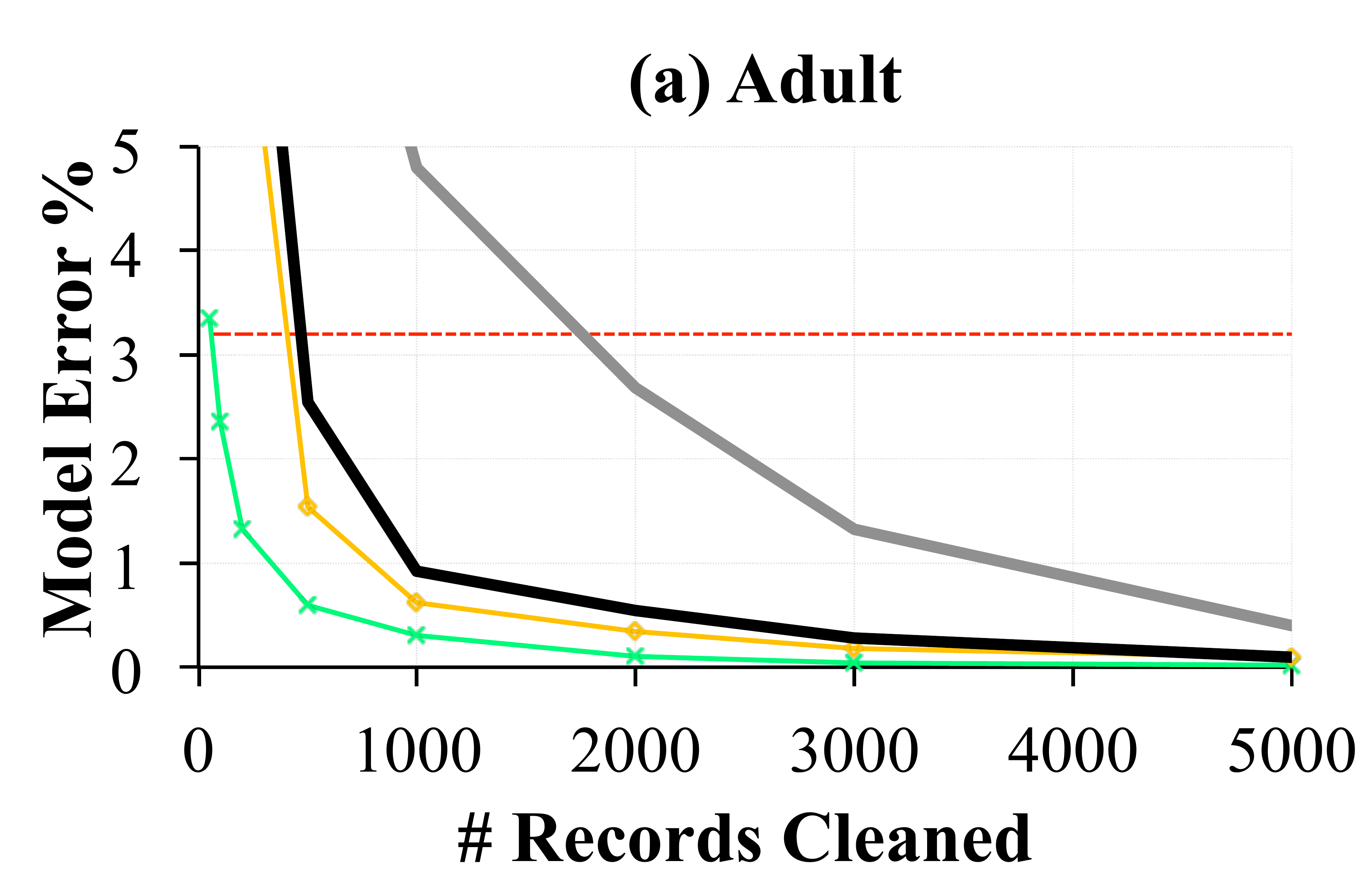}
    \includegraphics[width=0.49\columnwidth]{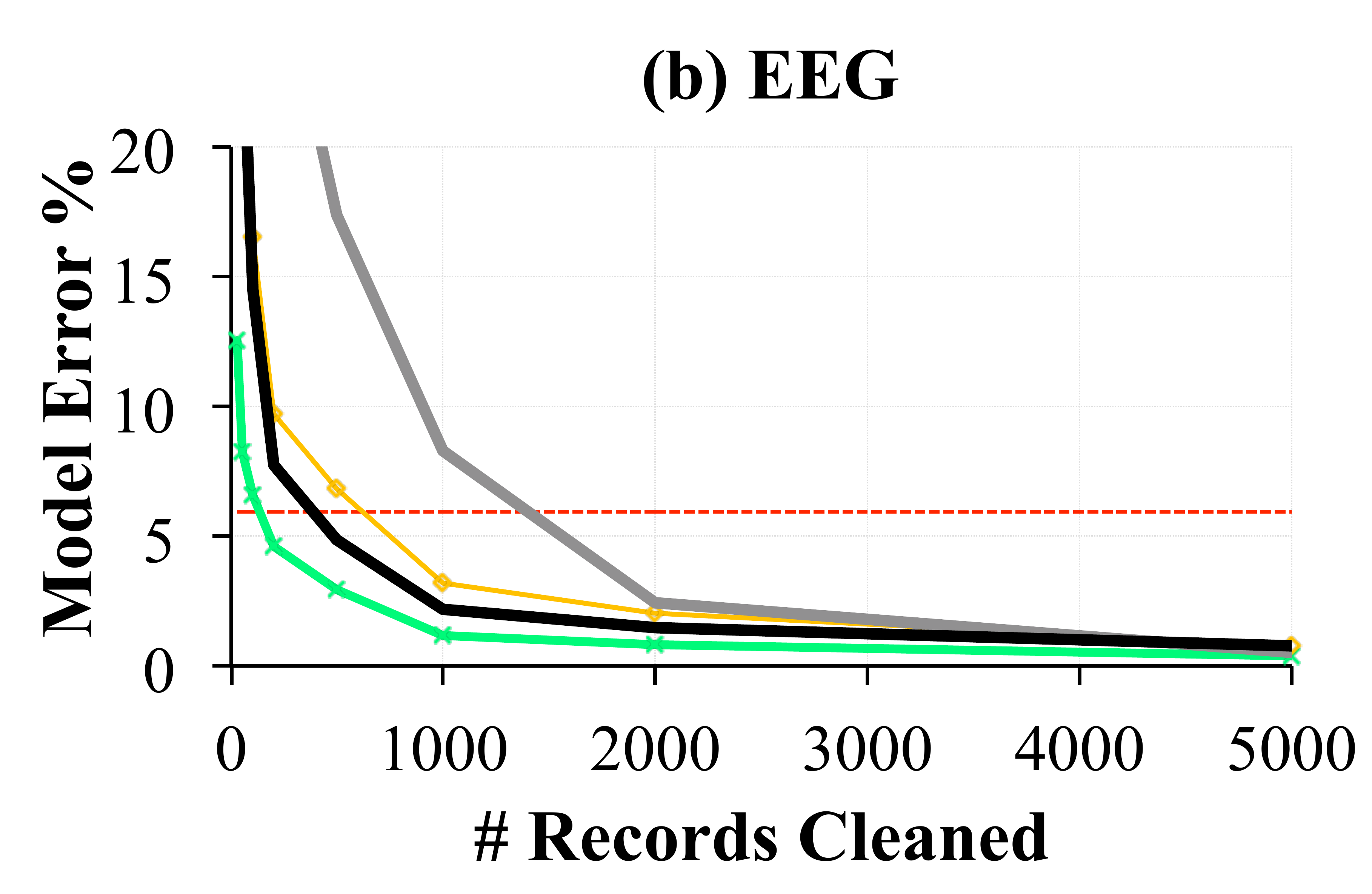}
    \includegraphics[width=0.49\columnwidth]{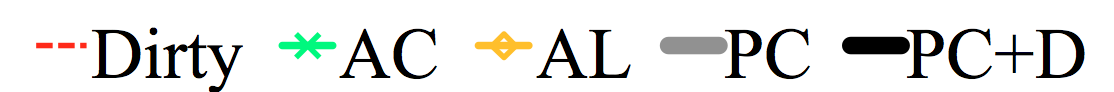}\vspace{-0.5em}
 \caption{The relative model error as a function of the number of examples cleaned. \sys converges with a smaller sample size to the true result in comparison to partial cleaning (PC,PC+D).  \label{pc-perf}}
\end{figure}

\noindent \emph{Summary: \sys converges faster than mixing dirty and clean data since it reweights data based on the fraction that is dirty and clean. Partial cleaning is not guaranteed to give sensible results.}

\vspace{1em}

\subsubsection{Corruption Rate}
The next experiment explores how much of the performance
is due to the initialization with the dirty model (i.e., SampleClean trains a model from ``scratch").
Figure \ref{bias} varies the systematic corruption rate and plots the number of records cleaned to achieve 1\% relative error for SampleClean and \sys.
SampleClean does not use the dirty data and thus its error is essentially governed by the Central Limit Theorem.
SampleClean outperforms \sys only when corruptions are very severe (45\% in Adult and nearly 60\% in EEG).
When the initialization with the dirty model is inaccurate, \sys does not perform as well. 

\begin{figure}[t]
\centering
 \includegraphics[width=0.49\columnwidth]{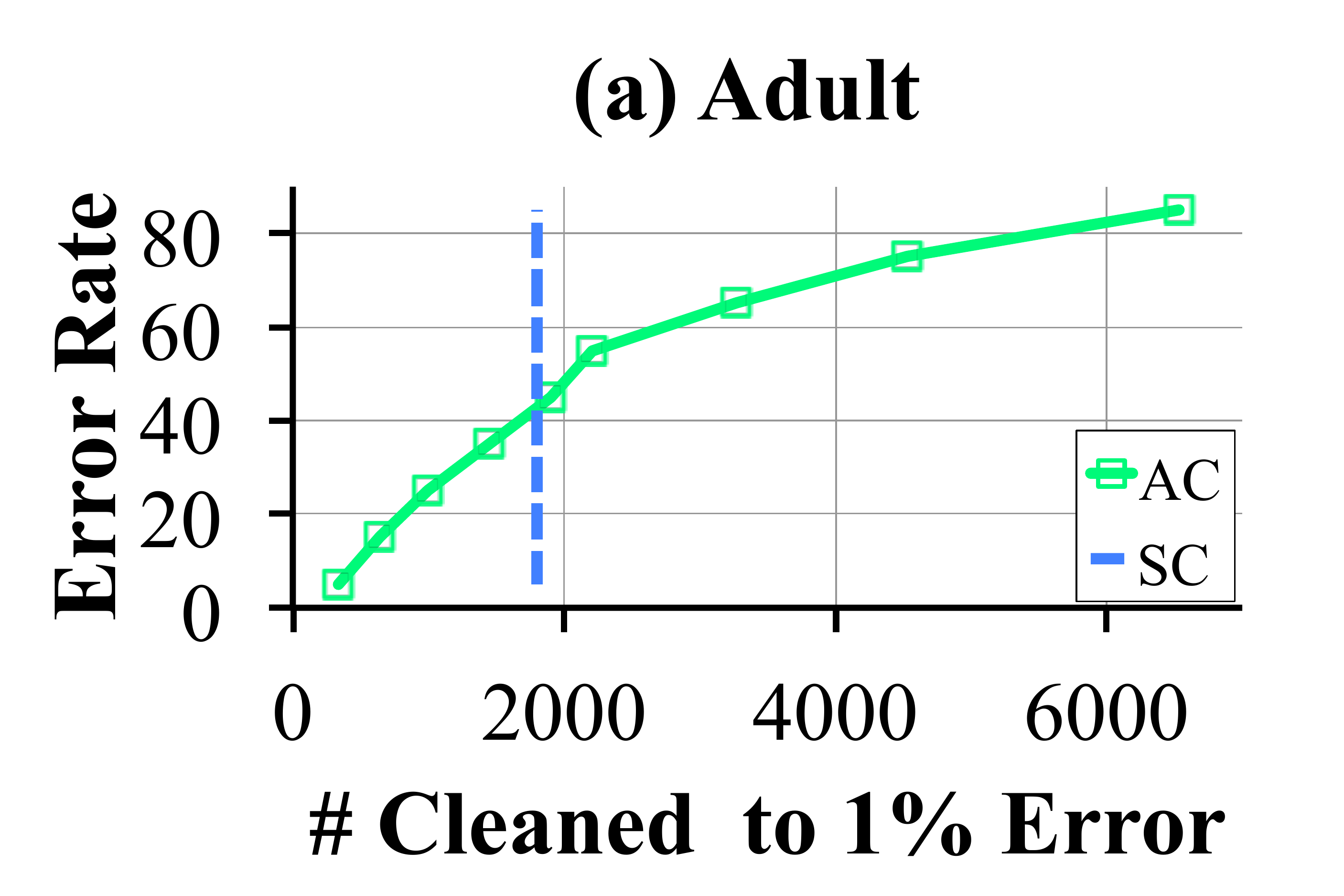}
  \includegraphics[width=0.49\columnwidth]{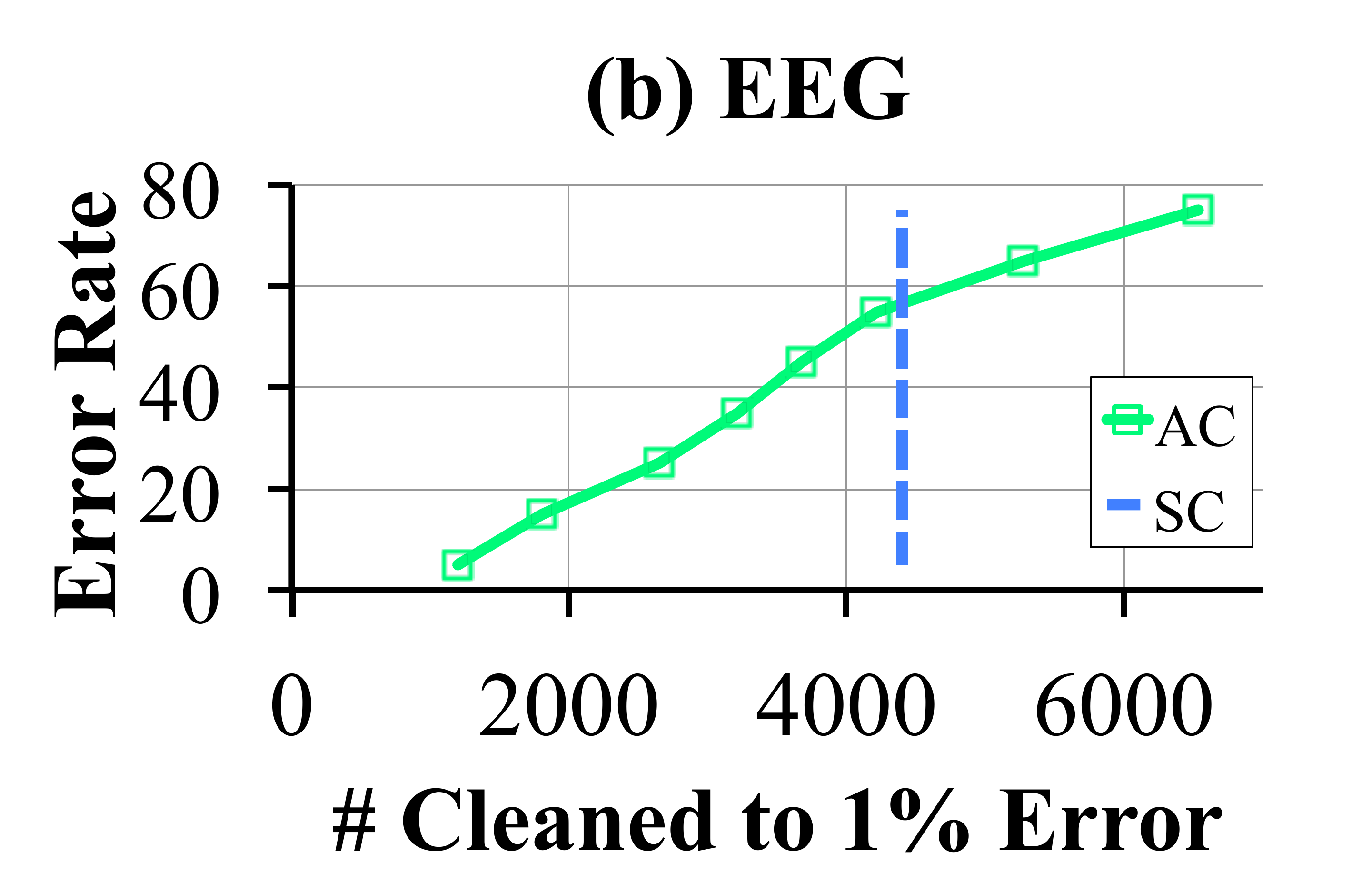}\vspace{-1em}
 \caption{\sys performs well until the corruption is so severe that the dirty model is not a good initialization. The error of SampleClean does not depend on the corruption rate so it is a vertical line.  \label{bias}}\vspace{-1.5em}
\end{figure}

\vspace{0.25em}

\noindent \emph{Summary: SampleClean is beneficial in comparison to \sys when corruption rates exceed 45\%.}

\subsection{\sys: Adaptive Detection}
This experiment explores how the results of the previous experiment change when using an adaptive detector instead of the \emph{a priori} detector.
Recall, in the systematic corruption, 3 of the most informative features were corrupted, thus we group these problems into $9$ classes.
We use an all-versus-one SVM to learn the categorization.

\subsubsection{Basic Performance}
Figure \ref{pred-perf} overlays the convergence plots in the previous experiments with a curve (denoted by AC+C) that represents \sys using a classifier instead of the \emph{a priori} detection. Initially \sys is comparable to Active Learning; however, as the classifier becomes more effective the detection improves the performance.
Over both datasets, at the 500 records point on the curve, adaptive \sys has a 30\% higher model error compared to \emph{a priori} \sys.
At 1000 records point on the curve, adaptive \sys has about 10\% higher error.

\vspace{0.25em}

\noindent \emph{Summary: For 500 records cleaned, adaptive \sys has a 30\% higher model error compared to a priori \sys, but still outperforms Active Learning and SampleClean.}

\begin{figure}[ht!]
\centering
 \includegraphics[width=0.49\columnwidth]{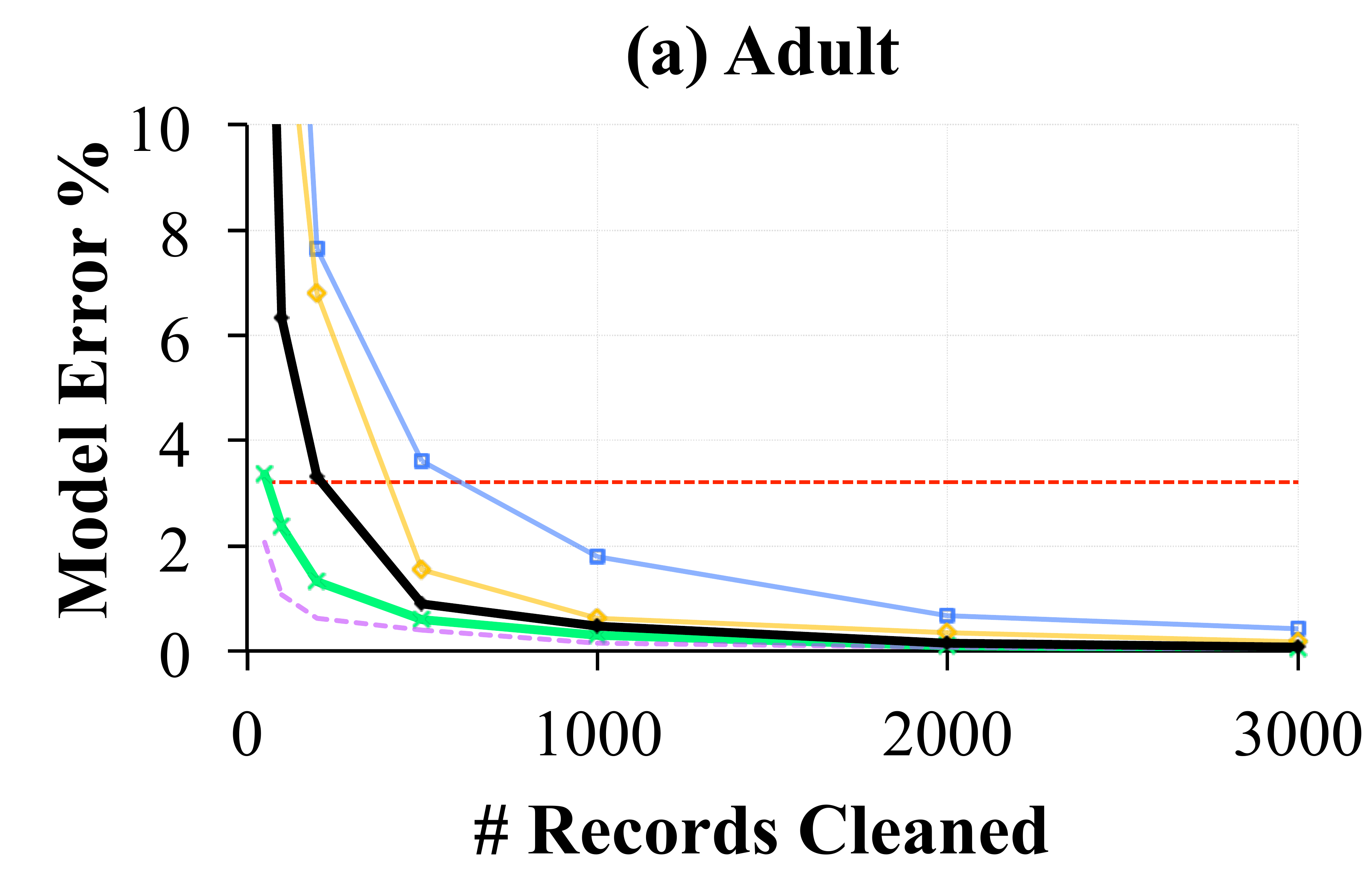}
 \includegraphics[width=0.49\columnwidth]{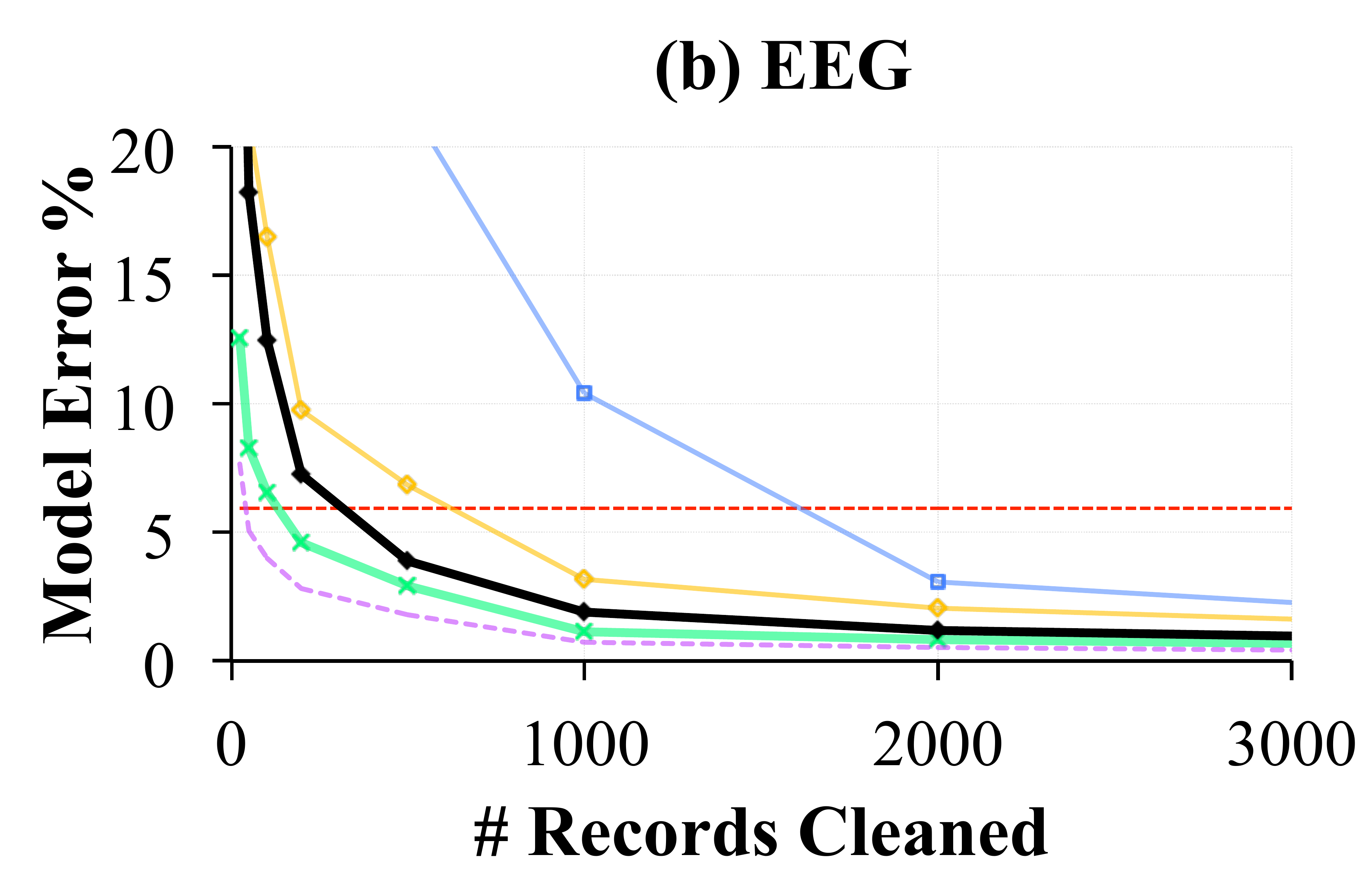}
 \includegraphics[width=0.49\columnwidth]{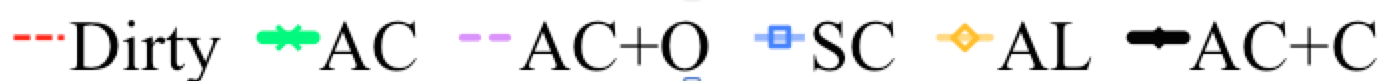}\vspace{-0.5em}
 \caption{Even with a classifier \sys converges faster than Active Learning and SampleClean. \label{pred-perf}}\vspace{-1.0em}
\end{figure}

\subsubsection{Classifiable Errors}
The adaptive case depends on being able to predict corrupted records.
For example, random corruption not correlated with any other data features may be hard to learn.
As corruption becomes more random, the classifier becomes increasingly erroneous.
The next experiment explores making the systematic corruption more random.
Instead of selecting the highest valued records for the most valuable features, we corrupt random records with probability $p$. 
We compare these results to AC-D where we do not have a detector at all at one vertical slice of the previous plot (cleaning 1000 records).
Figure \ref{tradeoffs2}a plots the relative error reduction using a classifier.
When the corruption is about 50\% random then there is a break even point where no detection is better.
The classifier is imperfect and misclassifies some data points incorrectly as cleaned.

\vspace{0.25em}

\noindent \emph{Summary: When errors are increasingly random (50\% random) and cannot be accurately classified, adaptive detection provides no benefit over no detection. }

\begin{figure}[ht!]
\centering
 \includegraphics[width=0.49\columnwidth]{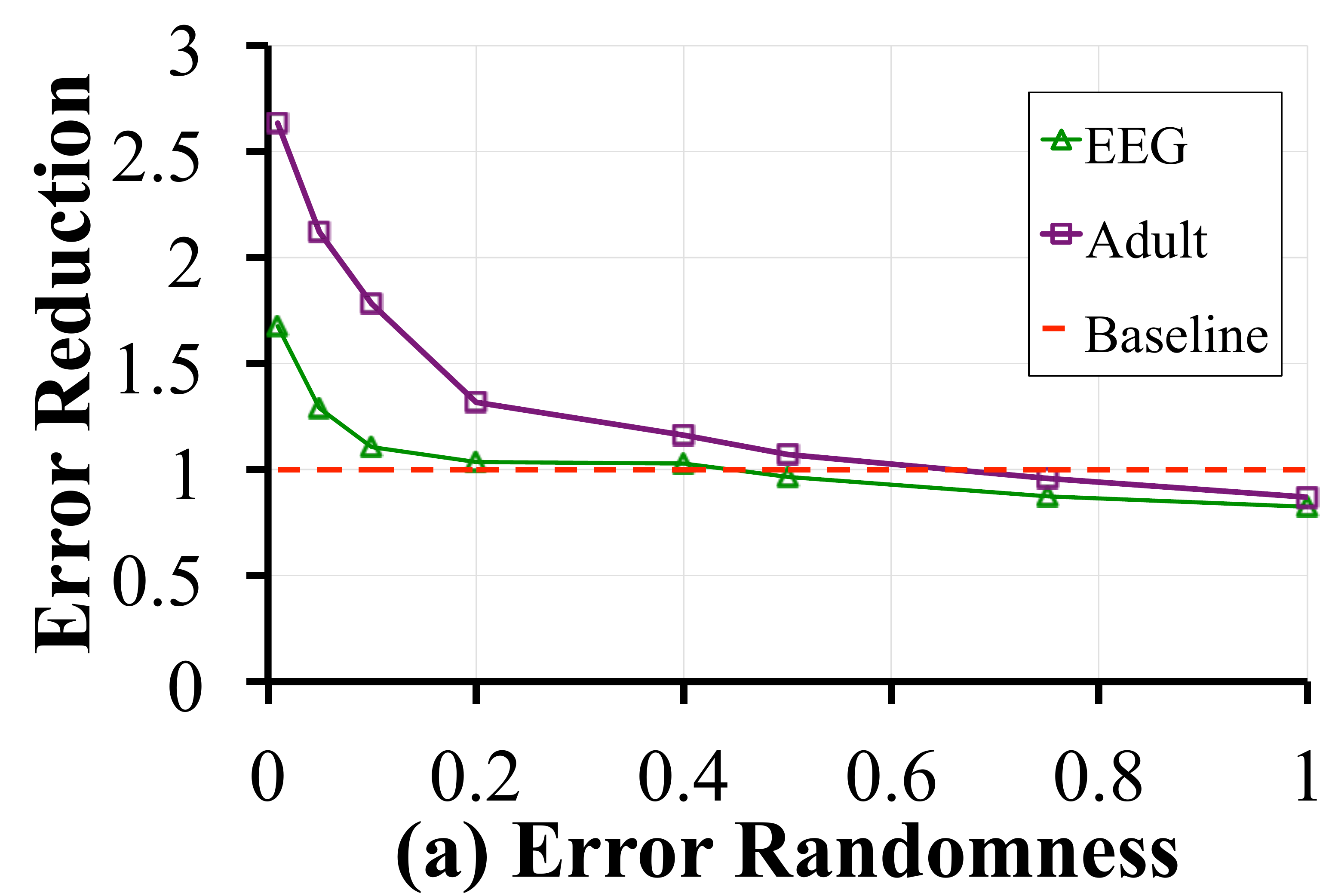}
 \includegraphics[width=0.49\columnwidth]{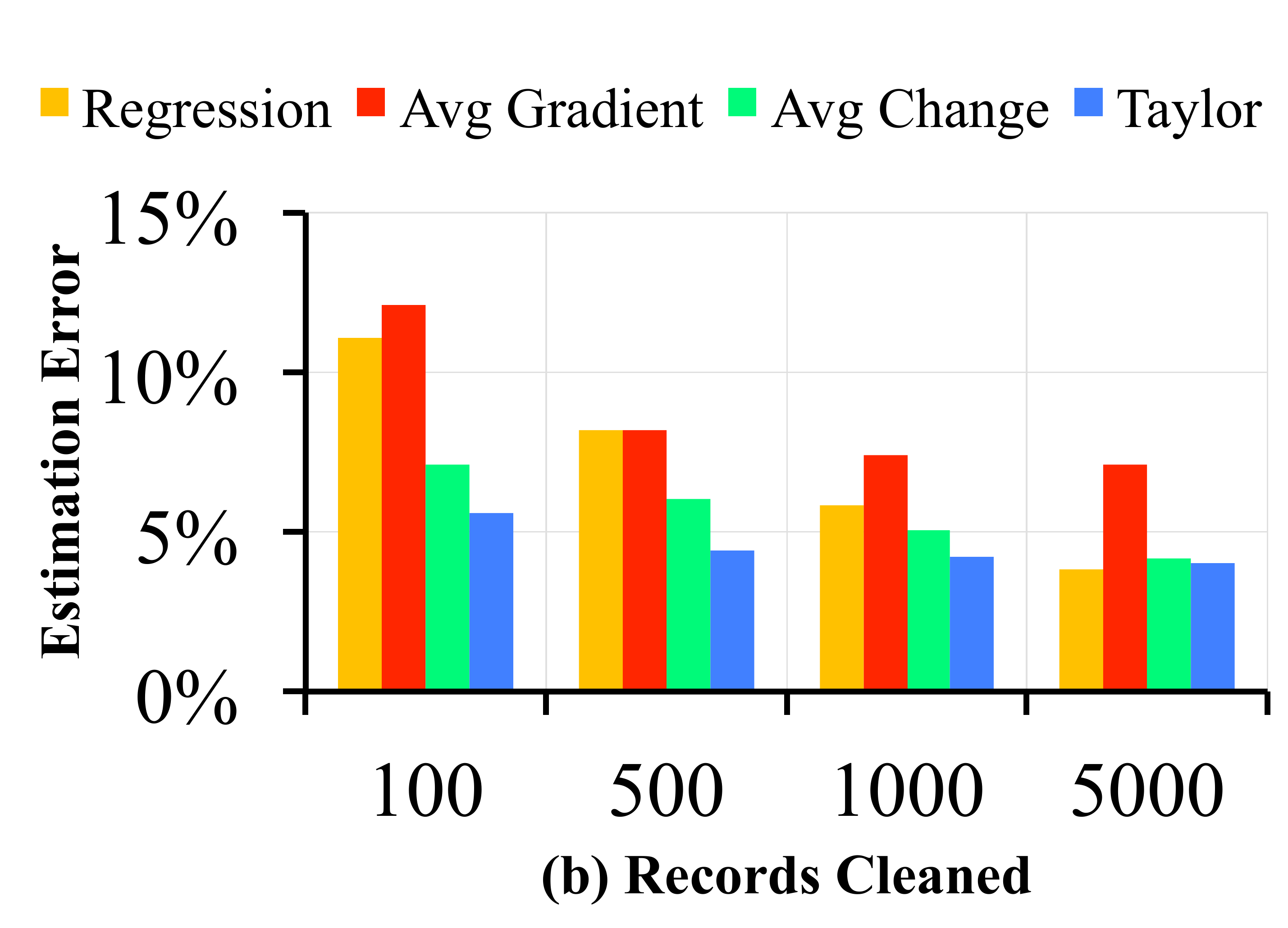}\vspace{-0.5em}
 \caption{(a) Data corruptions that are less random are easier to classify, and lead to more significant reductions in relative model error. (b) The Taylor series approximation gives more accurate estimates when the amount of cleaned data is small. \label{tradeoffs2}}
\end{figure}

\subsection{Estimation}\label{est}
The next experiment compares estimation techniques: (1) ``linear regression" trains a linear regression model that predicts the clean gradient as a function of the dirty gradient, (2) ``average gradient" which does not use the detection to inform how to apply the estimate, (3) ``average feature change" uses detection but no linearization, and (4) the Taylor series linear approximation.
Figure \ref{tradeoffs2}b measures how accurately each estimation technique estimates the gradient as a function of the number of cleaned records on the EEG dataset.

Estimation error is measured using the relative L2 error with the true gradient.
The Taylor series approximation proposed gives more accurate for small cleaning sizes.
Linear regression and the average feature change technique do eventually perform comparably but only after cleaning much more data.

\vspace{0.25em}

\noindent \emph{Summary: Linearized gradient estimates are more accurate when estimated from small samples. }

\subsection{Real World Scenarios}
The next set of experiments evaluate \sys in three real world scenarios, one demonstrating the \emph{a priori} case and the other two for the adaptive detection case.

\subsubsection{A Priori: Constraint Cleaning}\label{dfd-exp}
The first scenario explores the Dollars for Docs dataset published by ProPublica described throughout the paper.
To run this experiment, the entire dataset was cleaned up front, and simulated sampling from the dirty data and cleaning by looking up the value in the cleaned data (see Appendix \ref{dfd-errors} for constraints, errors, and cleaning methodology).
Figure \ref{dfd}a shows that \sys converges faster than Active Learning and SampleClean.
To achieve a 4\% relative error (i.e., a 75\% error reduction from the dirty model), \sys cleans 40000 fewer records than Active Learning.
Also, for 10000 records cleaned, \sys has nearly an order of magnitude smaller error than SampleClean.

Figure \ref{dfd}b shows the detection rate (fraction of disallowed research contributions identified) of the classifier as a function of the number of records cleaned. 
On the dirty data, we can only correctly classify 66\% of the suspected examples (88\% overall accuracy due to a class imbalance).
On the cleaned data, this classifier is nearly perfect with a 97\% true positive rate (98\% overall accuracy).
\sys converges to the cleaned accuracy faster than the alternatives with a classifier of 92\% true positive rate for only 10000 records cleaned.

\vspace{0.25em}

\noindent \emph{Summary: To achieve an 80\% detection rate, \sys cleans nearly 10x less records than Active Learning. }

\begin{figure}[t]
\centering\vspace{-1em}
 \includegraphics[width=0.49\columnwidth]{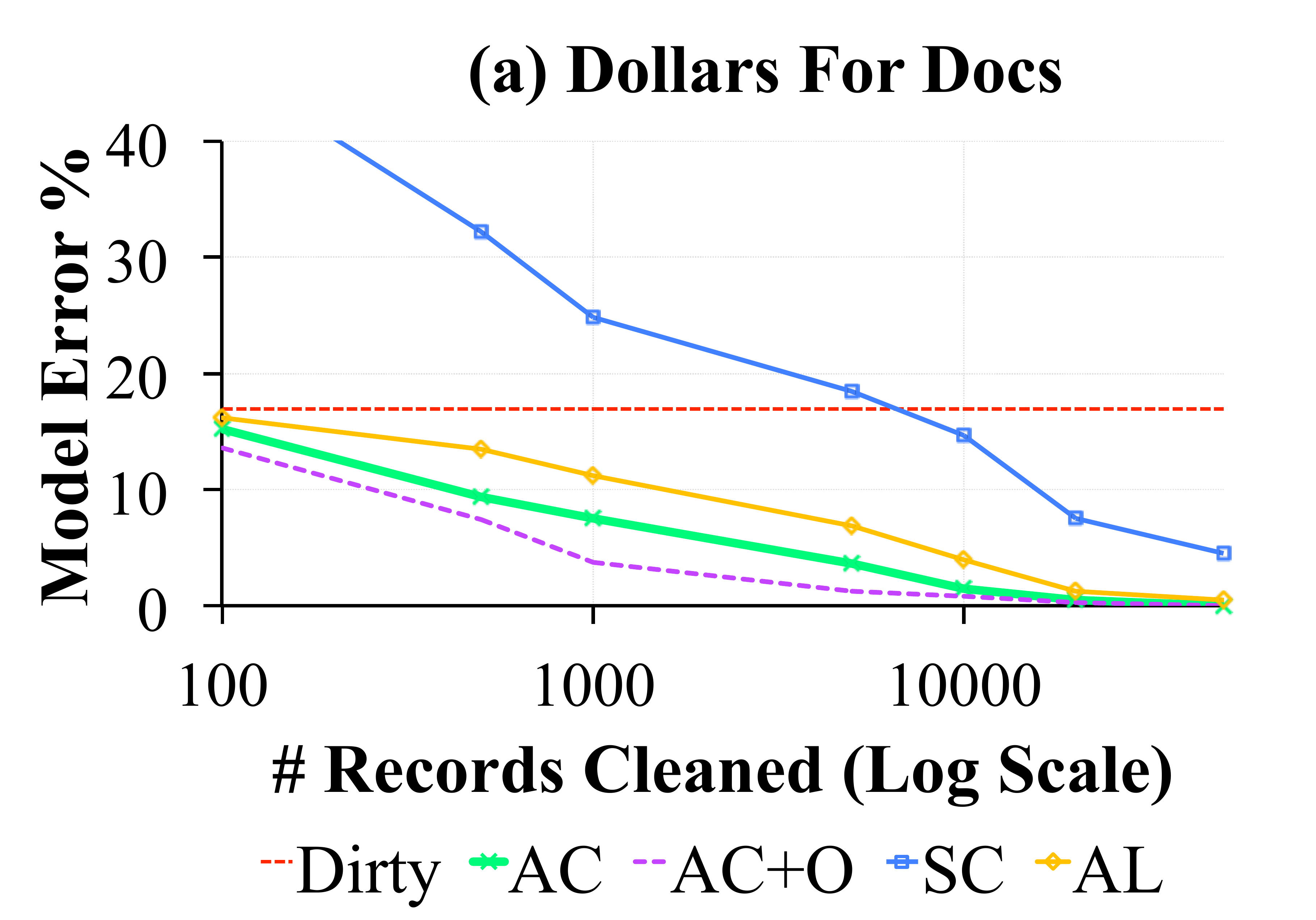}
 \includegraphics[width=0.49\columnwidth]{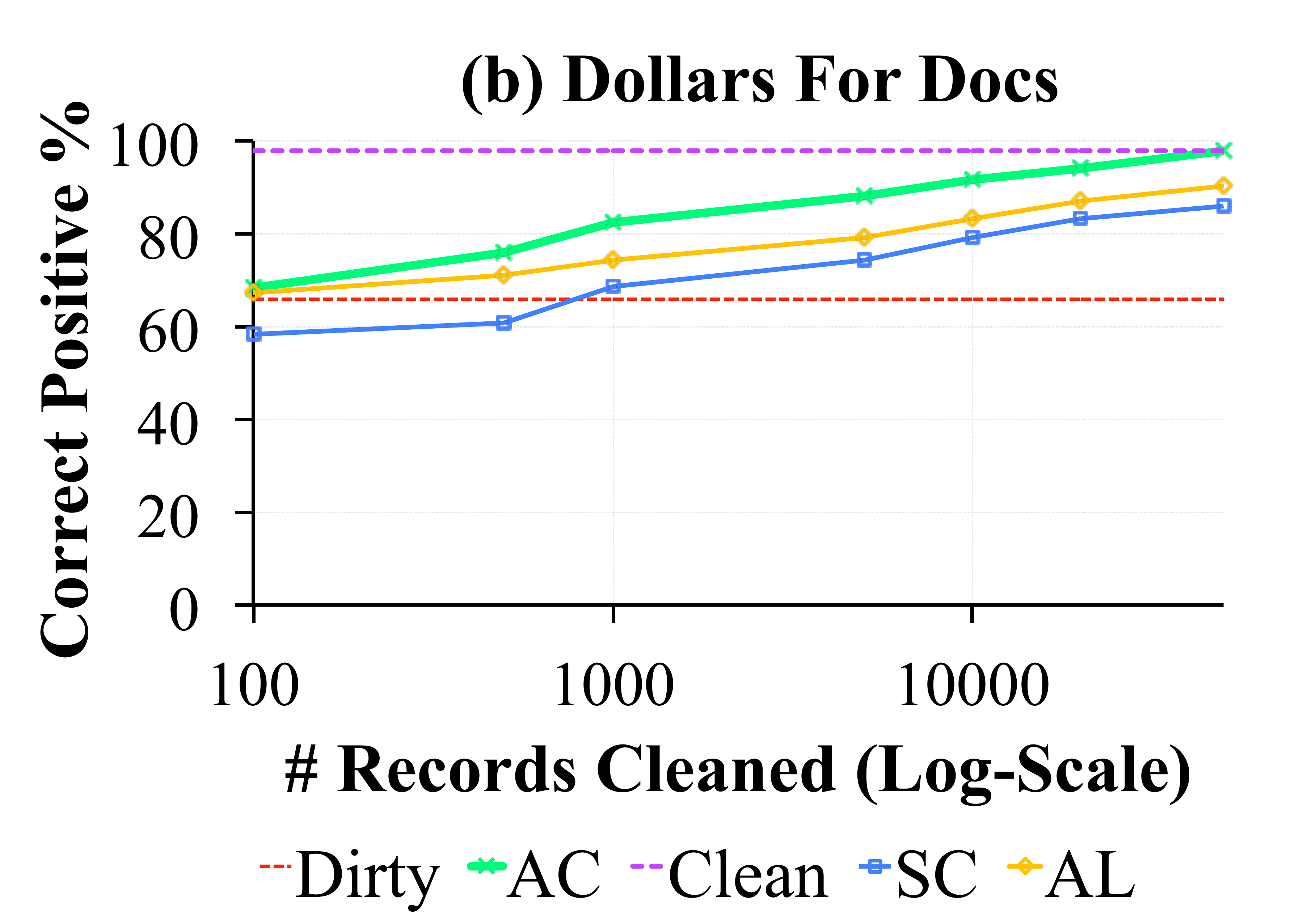}\vspace{-1em}
 \caption{(a) The relative model error as a function of the number of cleaned records. (b) The true positive rate as a function of the number of cleaned records. \label{dfd}}
\end{figure}

\subsubsection{Adaptive: Replacing Corrupted Data}
The next experiment explores the MNIST handwritten digit recognition dataset with a MATLAB image processing pipeline.
In this scenario, the analyst must inspect a potentially corrupted image and replace it with a higher quality one.
The MNIST dataset consists of 64x64 grayscale images.
There are two types of simulated corruptions: (1) 5x5 block removal where a random 5x5 block is removed from the image by setting its pixel values to 0, and (2) Fuzzy where a 4x4 moving average patch is applied over the entire image.
These corruptions are applied to a random 5\% of the images, and mimic the random (Fuzzy) vs. systematic corruption (5x5 removal) studied in the previous experiments.
The adaptive detector uses a 10 class classifier (one for each digit) to detect the corruption.

Figure \ref{mnist} shows that \sys makes more progress towards the clean model with a smaller number of examples cleaned.
To achieve a 2\% error for the block removal, \sys can inspect 2200 fewer images than Active Learning and 2750 fewer images than SampleClean.
For the fuzzy images, both Active Learning and \sys reach 2\% error after cleaning fewer than 100 images, while SampleClean requires 1750.

\vspace{0.25em}

\noindent \emph{Summary: In the MNIST dataset, \sys significantly reduces (more than 2x) the number of images to clean to train a model with 2\% error. }

\begin{figure}[t]
\centering
 \includegraphics[width=0.49\columnwidth]{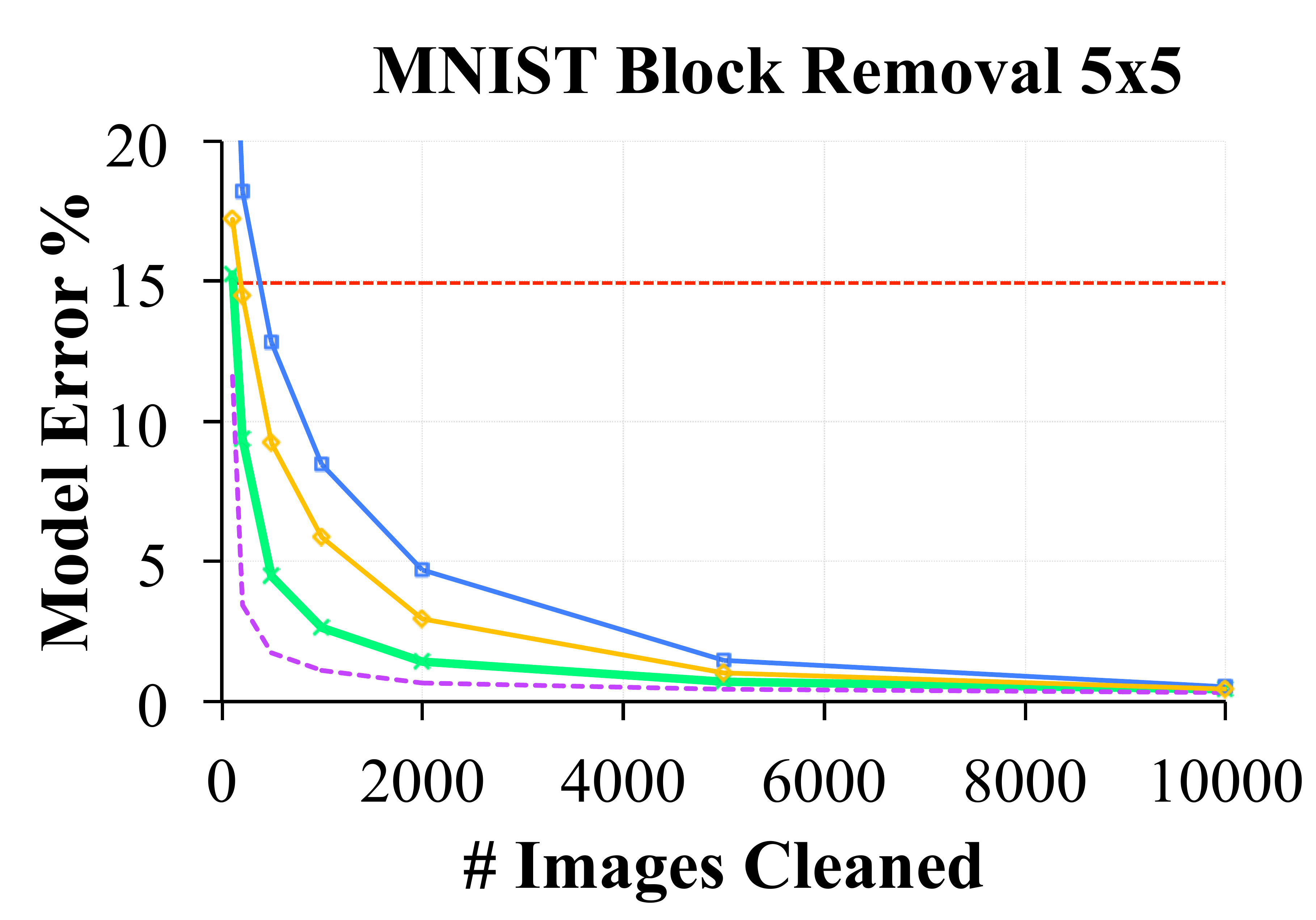}
 \includegraphics[width=0.49\columnwidth]{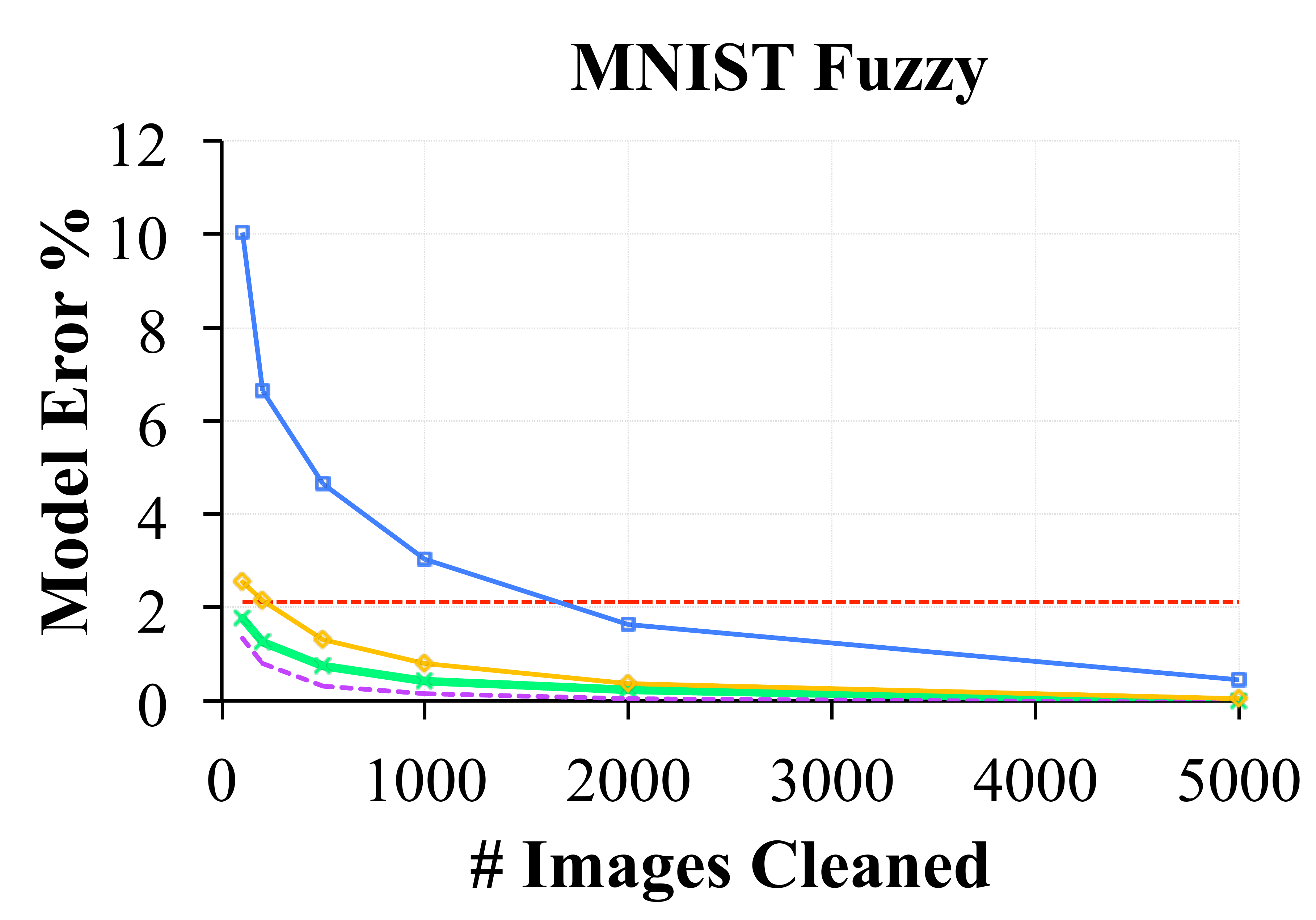}
 \includegraphics[width=0.49\columnwidth]{exp/legend-general.png}\vspace{-0.5em}
 \caption{In a real adaptive detection scenario with the MNIST dataset, \sys outperforms Active Learning and SampleClean.  \label{mnist}}\vspace{-1em}
\end{figure}

\subsubsection{Adaptive: Regression}
In the prior two experiments, we explored classification problems.
In this experiment, we consider the case when the convex model represents a linear regression model.
Regression models allow us to visualize what is happening when we apply \sys.
In Figure \ref{wb}, we illustrate regression model training on a small dataset of 193 countries collected from the World Bank.
Each country has an associated population and total dollar value of imports. 
We are interested in examining the relationship between these variables.
However, for some countries, the import values are out-of-date in the World Bank dataset. 
Up-to-date values are usually avaiable on national statistics websites and can be determined with some web searching.
It turns out that smaller countries were more likely to have out-of-date statistics in the World Bank dataset, and as a result, the trend line is misleading in the dirty data.
We applied \sys after verifying 30 out of the 193 countries (marked in yellow), and found that we could achieve a highly accurate approximation of the full result.

\begin{figure}[t]
\centering
 \includegraphics[width=0.49\columnwidth]{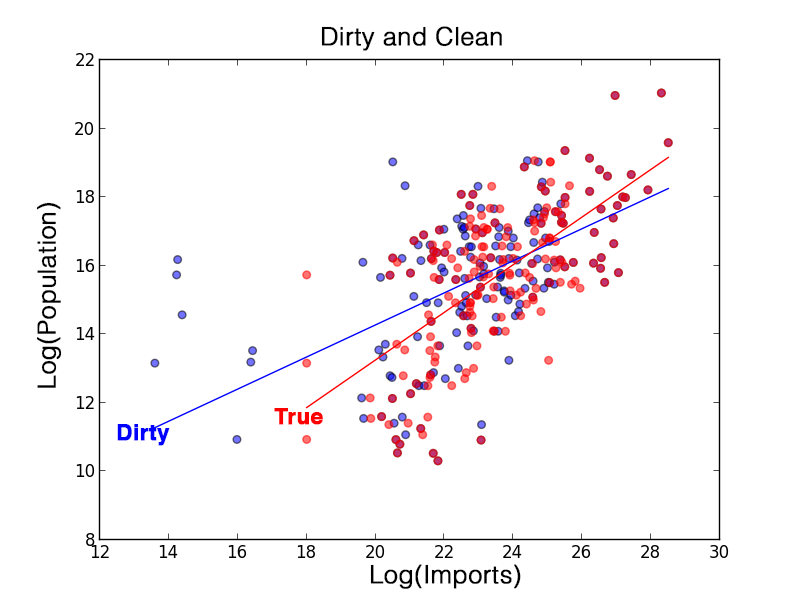}
 \includegraphics[width=0.49\columnwidth]{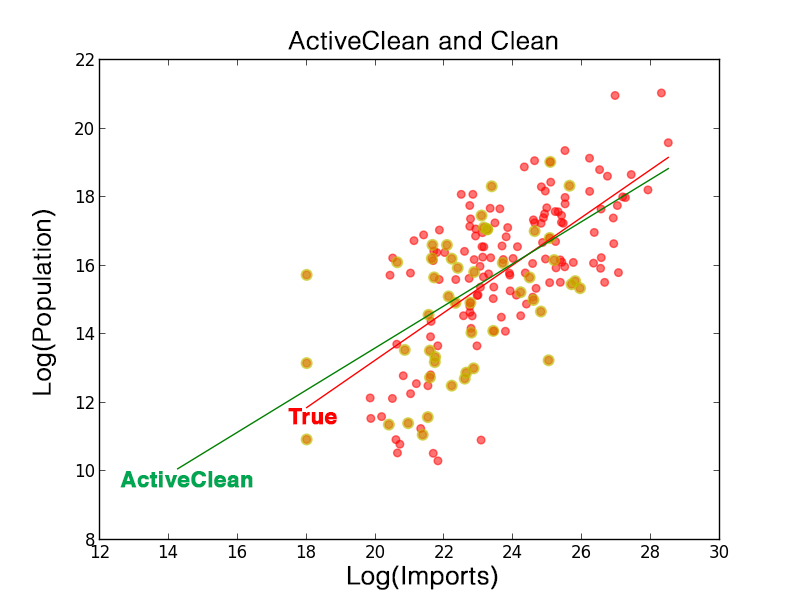}
 \caption{World Bank Data. We apply \sys to learn an accurate model predicting population from import values. The data has a systematic bias where small countries have out-of-date import values. \label{wb}}\vspace{-1em}
\end{figure}

\vspace{0.25em}

\noindent \emph{Summary: \sys is accurate even in regression analytics. }

%% file: relatedwork.tex
\vspace{-1em}
\section{Related Work}\label{rw}
\noindent \textbf{Data Cleaning: } 
When data cleaning is expensive, it is desirable to apply it \textbf{progressively}, where analysts can inspect early results with only $k \ll N$ records cleaned.
Progressive data cleaning is a well studied problem especially in the context of entity resolution \cite{altowim2014progressive, whang2014incremental, papenbrock2015progressive, gruenheid2014incremental}.
Prior work has focused on the problem of designing data structures and algorithms to apply data cleaning progressively. which is challenging because many data cleaning algorithms require information from the entire relation.
Over the last 5 years a number of new results have expanded the scope and practicality of progressive data cleaning~\cite{mayfield2010eracer, DBLP:journals/pvldb/YakoutENOI11, yakout2013don}.
\sys studies the problem of prioritizing progressive cleaning by leveraging information about a user's subsequent use for the data.
Certain records, if cleaned, may be more likely to affect the downstream analysis.

There are a number of other works that use machine learning to improve the efficiency and/or reliability of data cleaning~\cite{DBLP:journals/pvldb/YakoutENOI11,yakout2013don,gokhale2014corleone}.
For example, Yakout et al. train a model that evaluates the likelihood of a proposed replacement value \cite{yakout2013don}.
Another application of machine learning is value imputation, where a missing value is predicted based on those records without missing values.
Machine learning is also increasingly applied to make automated repairs more reliable with human validation \cite{DBLP:journals/pvldb/YakoutENOI11}.
Human input is often expensive and impractical to apply to entire large datasets.
Machine learning can extrapolate rules from a small set of examples cleaned by a human (or humans) to uncleaned data \cite{gokhale2014corleone, DBLP:journals/pvldb/YakoutENOI11}.
This approach can be coupled with active learning \cite{DBLP:journals/pvldb/MozafariSFJM14} to learn an accurate model with the fewest possible number of examples.
While, in spirit, \sys is similar to these approaches, it addresses a very different problem of data cleaning before user-specified modeling.
The key new challenge in this problem is ensuring the correctness of the user's model after partial data cleaning.

SampleClean~\cite{wang1999sample} applies data cleaning to a sample of data, and estimates the results of aggregate queries.
Sampling has also been applied to estimate the number of duplicates in a relation \cite{heise2014estimating}. 
Similarly, Bergman et al. explore the problem of query-oriented data cleaning \cite{DBLP:conf/sigmod/BergmanMNT15}, where given a query, they clean data relevant to that query. 
Existing work does not explore cleaning driven by the downstream machine learning ``queries" studied in this work.
Deshpande et al. studied data acquisition in sensor networks \cite{deshpande2004model}. They explored value of information based prioritization of data acquisition for estimating aggregate queries of sensor readings.
Similarly, Jeffery et al. \cite{DBLP:conf/pervasive/JefferyAFHW06} explored similar prioritization based on value of information.
We see this work as pushing prioritization further down the pipeline to the end analytics.
Finally, incremental optimization methods like SGD have a connection to incremental materialized view maintenance as the argument for incremental maintenance over recomputation is similar (i.e., relatively sparse updates).
Krishnan et al. explored how samples of materialized views can be maintained similar to how models are updated with a sample of clean data in this work \cite{krishnan2015svc}.

\vspace{0.5em}

\noindent \textbf{Stochastic Optimization and Active Learning: } Zhao and Tong recently proposed using importance sampling in conjunction with stochastic gradient descent \cite{zhao2014stochastic}. 
The ideas applied in \sys are well rooted in the Machine Learning and Optimization literature, and we apply these ideas to the data cleaning problem. 
This line of work builds on prior results in linear algebra that show that some matrix columns are more informative than others \cite{drineas2012fast}, and Active Learning which shows that some labels are more informative that others \cite{settles2010active}.
Active Learning largely studies the problem of label acquisition \cite{settles2010active},
and recently the links between Active Learning and Stochastic optimization have been studied \cite{guillory2009active}. 
We use the work in Guillory et al. to evaluate a state-of-the-art Active Learning technique against \sys.

\vspace{0.5em}

\noindent \textbf{Transfer Learning and Bias Mitigation: }  
\sys has a strong link to a field called Transfer Learning and Domain Adaptation \cite{pan2010survey}. The basic idea of Transfer Learning is that suppose a model is trained on a dataset $D$ but tested on a dataset $D'$. 
Much of the complexity and contribution of \sys comes from efficiently tuning such a process for expensive data cleaning applications -- costs not studied in Transfer Learning.
In robotics, Mahler et al. explored a calibration problem in which data was systematically corrupted \cite{DBLP:conf/case/MahlerKLSMKPWFAG14} and proposed a rule-based technique for cleaning data.
Other problems in bias mitigation (e.g., Krishnan et al. \cite{DBLP:conf/recsys/KrishnanPFG14}) have the same structure, systematically corrupted data that is feeding into a model.
In this work, we try to generalize these principles given a general dirty dataset, convex model, and data cleaning procedure.

\vspace{0.5em}

\noindent \textbf{Secure Learning: } \sys is also related to work in adversarial learning \cite{nelson2012query}, where the goal is to make models robust to adversarial data manipulation.
This line of work has extensively studied methodologies for making models private to external queries and robust to malicious labels \cite{xiaofeature}, but the data cleaning problem explores more general corruptions than just malicious labels.
One widely applied technique in this field is reject-on-negative impact, which essentially, discards data that reduces the loss function--which will not work when we do not have access to the true loss function (only the ``dirty loss").

%% file: discussion.tex
\section{Discussion and Future Work}
The experimental results suggest the following conclusions about \sys: (1) when the data corruption rate is relatively small (e.g., 5\%), \sys cleans fewer records than Active Learning or SampleClean to achieve the same model accuracy, (2) all of the optimizations in \sys (importance sampling, detection, and estimation) lead to significantly more accurate models at small sample sizes, (3) only when corruption rates are very severe (e.g. 50\%) , SampleClean outperforms \sys, and (4) two real-world scenarios demonstrate similar accuracy improvements where \sys returns significantly more accurate models than SampleClean or Active Learning for the same number of records cleaned.

There are also a few additional points for discussion.
\sys provides guarantees for training error on models trained with progressive data cleaning, however, there are no such guarantees on test error. 
This work focuses on the problem where an analyst has a large amount of dirty data and would like explore data cleaning and predictive models on this dataset.
By providing the analyst more accurate model estimates, the value of different data cleaning techniques can be judged without having to clean the entire dataset.
However, the exploratory analysis problem is distinct from the model deployment problem (i.e., serving predictions to users from the model), which we hope to explore in more detail in future work.
It implicitly assumes that when the model is deployed, it will be applied in a setting where the test data is also clean.
Training on clean data, and testing on dirty data, defeats the purpose of data cleaning and can lead to unreliable predictions.

As the experiments clearly show, \sys is not strictly \emph{better} than Active Learning or SampleClean.
\sys is optimized for a specific design point of sparse errors and small sample sizes, and the empirical results suggest it returns more accurate models in this setting.
As sample sizes and error rates increase, the benefits of \sys are reduced.
Another consideration for future work is automatically selecting alternative techniques when \sys is expected to perform poorly.

Beyond these limitations, there are several exciting new avenues for future work.
The data cleaning models explored in this work can be extended to handle non-uniform costs, where different errors have a different cleaning cost.
Next, the empirical success of Deep Learning has led to increasing industry and research adoption of non-convex losses in many tasks that were traditionally served by convex models.
In future work, we hope to explore how we can integrate with such frameworks.

%% file: conclusion.tex
\section{Conclusion}
The growing popularity of predictive models in data analytics adds additional challenges in managing dirty data.
Progressive data cleaning in this setting is susceptible to errors due to mixing dirty and clean data, sensitivity to sample size, and the sparsity of errors.
The key insight of \sys is that an important class of predictive models, called convex loss models (e.g., linear regression and SVMs), can be simultaneously trained and cleaned.
Consequently, there are provable guarantees on the convergence and error bounds of \sys.  
\sys also includes numerous optimizations such as: using the information from the model to inform data cleaning on samples, dirty data detection to avoid sampling clean data, and batching updates.
The experimental results are promising as they suggest that these optimizations can significantly reduce data cleaning costs when errors are sparse and cleaning budgets are small.
Techniques such as Active Learning and SampleClean are not optimized for the sparse low-budget setting, and \sys achieves models of similar accuracy for significantly less records cleaned.

%% file: appendix.tex
\section{Set-of-Records Cleaning Model}\label{set-of-r}
In paper, we formalized the analyst-specified data cleaning as follows.
We take the sample of the records $S_{dirty}$, and apply data cleaning $C(\cdot)$.
$C$ is applied to a record and produces the clean record:
\[
S_{clean} = \{C(r) : \forall r \in S_{dirty}\}
\]
The record-by-record cleaning model is a formalization of the costs of data cleaning where each record has the same cost to clean and this cost does not change throughout the entire cleaning session.
There are, however, some cases when cleaning the first record of a certain type of corruption is expensive but all subsequent records are cheaper.

\begin{example}\label{app-ex1}
In most spell checking systems, when a misspelling is identified, the system gives an option to fix all instances of that misspelling.
\end{example}

\begin{example}\label{app-ex2}
When an inconsistent value is identified all other records with the same inconsistency can be efficiently fixed.
\end{example}

This model of data cleaning can fit into our framework and we formalize it as the ``Set-of-Records" model as opposed to the ``Record-by-Record" model. 
In this model, the cleaning function $C(\cdot)$ is not restricted to updating only the records in the sample.
$C(\cdot)$ takes the entire dirty sample as an argument (that is the cleaning is a function of the sample), the dirty data, and updates the entire dirty data:
\[
R'_{dirty} = C(S_{dirty},R_{dirty})
\]
we require that for every record $s \in S_{dirty}$, that record is completely cleaned after applying $C(\cdot)$, giving us $S_{clean}$.
Records outside of $S_{dirty}$ may be cleaned on a subset of dirty attributes by $C(\cdot)$.
After each iteration, we re-run the detector, and move any $r \in R'_{dirty}$ that are clean to $R_{clean}$.
Such a model allows us to capture data cleaning operations such as in Example \ref{app-ex1} and Example \ref{app-ex2}.

\section{Stochastic Gradient Descent}\label{appsgd}

Stochastic Gradient Descent converges for a suitably chosen step size if the sample gradients are unbiased estimates of the full gradient. 
The first problem is to choose weights $\alpha$ and $\beta$ (to average already clean and newly cleaned data) such that the estimate of the gradient is unbiased. 
The batch $S_{dirty}$ is drawn only from $R_{dirty}$.
Since the sizes of $R_{dirty}$ and its complement are known, it follows that the gradient over the already clean data $g_C$ and the recently cleaned data $g_S$ can be combined as follows:
\[
g(\theta^{t}) = \frac{\mid R_{dirty} \mid \cdot g_S + \mid R_{clean} \mid \cdot g_C  }{\mid R \mid}
\]
Therefore,
\[
\alpha = \frac{\mid R_{clean} \mid}{\mid R \mid}, \beta = \frac{\mid R_{dirty} \mid}{\mid R \mid}
\]

\begin{lemma}
The gradient estimate $g(\theta)$ is unbiased if $g_S$ is an unbiased estimate of:
\[
\frac{1}{\mid R_{dirty} \mid} \sum g_i(\theta)
\]
\end{lemma}
\begin{proof}[Sketch]
\[
\mathbb{E}(\frac{1}{\mid R_{dirty} \mid} \sum g_i(\theta)) = \frac{1}{\mid R_{dirty} \mid} \cdot \mathbb{E}(\sum g_i(\theta)))
\]
By symmetry, 
\[
\mathbb{E}(\frac{1}{\mid R_{dirty} \mid} \sum g_i(\theta)) = g(\theta)
\]
\[
\mathbb{E}(\frac{1}{\mid R_{dirty} \mid} \sum g_i(\theta)) = \frac{\mid R_{dirty} \mid \cdot g_S + \mid R_{clean} \mid \cdot g_C  }{\mid R \mid}
\]
\end{proof}

The error bound discussed in Proposition 2 can be tightened for a class of models called strongly convex (see \cite{bertsekas2011incremental} for a defintion). 

\begin{proposition}
For a strongly convex loss, a batch size $b$, and $T$ iterations, the convergence rate is bounded by $O(\frac{\sigma^2}{bT})$. 
\end{proposition}

\section{Non-convex losses}\label{non-convex}
We acknowledge that there is an increasing popularity of non-convex losses in the Neural Network and Deep Learning literature. 
However, even for these losses, gradient descent techniques still apply. 
Instead of converging to a global optimum they converge to a locally optimal value. 
Likewise, \sys will converge to the closest locally optimal value to the dirty model. 
Because of this, it is harder to reason about the results.
Different initializations will lead to different local optima, and thus, introduces a complex dependence on the initialization with the dirty model.
This problem is not fundemental to \sys and any gradient technique suffers this challenge for general non-convex losses, and we hope to explore this more in the future.

\section{Importance Sampling}\label{impsample-deriv}
This lemma describes the optimal distribution over a set of scalars:
\begin{lemma}\label{impsample}
Given a set of real numbers $A = \{a_1,...,a_n\}$, let $\hat{A}$ be 
a sample with replacement of $A$ of size k.
If $\mu$ is the mean $\hat{A}$, the sampling distribution that minimizes
 the variance of $\mu$, i.e., the expected square error, is $p(a_i) \propto a_i$.
\end{lemma}
Lemma \ref{impsample} shows that when estimating a mean of numbers with sampling, the distribution with optimal variance is sampling proportionally to the values.

The variance of this estimate is given by:
\[
Var(\mu) = \mathbb{E}(\mu^2)-\mathbb{E}(\mu)^2
\] 
Since the estimate is unbiased, we can replace $\mathbb{E}(\mu)$ with the average of $A$:
\[
Var(\mu) = \mathbb{E}(\mu^2)-\bar{A}^2
\]
Since $\bar{A}$ is deterministic, we can remove that term during minimization.
Furthermore, we can write $\mathbb{E}(\mu^2)$ as:
\[
\mathbb{E}(\mu^2) = \frac{1}{n^2}\sum_i^n \frac{a_i^2}{p_i}
\]
Then, we can solve the following optimization problem (removing the proportionality of $\frac{1}{n^2}$) over the set of weights $P=\{p(a_i)\}$:
\[
\min_{P} \sum_i^N \frac{a_i^2}{p_i}
\]
\[
\text{subject to: } P > 0, \sum P = 1
\]
Applying Lagrange multipliers, an equivalent unconstrained optimization problem is:
\[
\min_{P > 0,\lambda > 0} \sum_i^N \frac{a_i^2}{p_i} + \lambda \cdot (\sum P - 1)
\]
If, we take the derivatives with respect to $p_i$ and set them equal to zero:
\[
-\frac{a_i^2}{2 \cdot p_i^2} + \lambda = 0
\]
If, we take the derivative with respect to $\lambda$ and set it equal to zero:
\[
\sum P - 1
\]
Solving the system of equations, we get:
\[
p_i = \frac{\mid a_i \mid }{\sum_i \mid a_i \mid}
\]

\section{Linearization}\label{apptaylor}
If $d$ is the dirty value and $c$ is the clean value, the Taylor series approximation for a function $f$ is given as follows:
\[
f(c) = f(d) + f'(d)\cdot(d-c) + ...
\]
Ignoring the higher order terms, the linear term $f'(d)\cdot(d-c)$ is a linear function in each feature and label.
We only have to know the change in each feature to estimate the change in value.
In our case the function $f$ is the gradient $\nabla\phi$.
So, the resulting linearization is:
\[
\nabla\phi(x^{(c)}_i,y^{(c)}_i,\theta) \approx \nabla\phi(x,y,\theta) + \frac{\partial}{\partial X}\nabla\phi(x,y,\theta)\cdot (x - x^{(c)}) \]
\[+ \frac{\partial}{\partial Y}\phi(x,y,\theta)\cdot (y - y^{(c)})
\]
When we take the expected value:
\[
\mathbb{E}(\nabla\phi(x_{clean},y_{clean},\theta)) \approx \nabla\phi(x,y,\theta) + \frac{\partial}{\partial X}\nabla\phi(x,y,\theta)\cdot \mathbb{E}(\Delta x) \]
\[+ \frac{\partial}{\partial Y}\nabla\phi(x,y,\theta)\cdot \mathbb{E}(\Delta y)
\]
It follows that:
\[
\approx \nabla\phi(x,y,\theta) + M_x \cdot \mathbb{E}(\Delta x) + M_y \cdot \mathbb{E}(\Delta y)
\]
where $M_x = \frac{\partial}{\partial X}\nabla\phi$ and $M_y = \frac{\partial}{\partial Y}\nabla\phi$.
Recall that the feature space is $d$ dimensional and label space is $l$ dimensional.
Then, $M_x$ is an $d \times d$ matrix, and $M_y$ is a $d \times l$ matrix.
Both of these matrices are computed for each record.
$\Delta x$ is a $d$ dimensional vector where each component represents a change in that feature and $\Delta y$ is an $l$ dimensional vector that represents the change in each of the labels. 

This linearization allows \sys to maintain per feature (or label) average changes and use these changes to center the optimal sampling distribution around the expected clean value.
To estimate $\mathbb{E}(\Delta x)$ and $\mathbb{E}(\Delta y)$, consider the following for a single feature $i$:
If we average all $j=\{1,...,K\}$ records cleaned that have an error for that feature, weighted by their sampling probability:
\[
\bar{\Delta}_{xi} = \frac{1}{NK}\sum_{j=1}^K (x^{(d)}[i]-x^{(c)}[i])\times \frac{1}{p(j)}
\]
Similarly, for a label $i$:
\[
\bar{\Delta}_{yi} = \frac{1}{NK}\sum_{j=1}^K (y^{(d)}[i]-y^{(c)}[i])\times \frac{1}{p(j)}
\]

Each $\bar{\Delta}_{xi}$ and $\bar{\Delta}_{yi}$ represents an average change in a single feature.
A single vector can represent the necessary changes to apply to a record $r$:
For a record $r$, the set of corrupted features is $f_r,l_r$.
Then, each record $r$ has a d-dimensional vector $\Delta_{rx}$ which is constructed as follows:
\[
 \Delta_{rx}[i] = \begin{cases} 0 & i \notin f_r \\ 
\bar{\Delta}_{xi} & i \in f_r
\end{cases} 
\]
Each record $r$ also has an l-dimensional vector $\Delta_{ry}$ which is constructed as follows:
\[
 \Delta_{rx}[i] = \begin{cases} 0 & i \notin l_r \\ 
\bar{\Delta}_{yi} & i \in l_r
\end{cases} 
\]
Finally, the result is: 
\[p_{r}\propto\|\nabla\phi(x,y,\theta^{(t)}) + M_x \cdot \Delta_{rx} +  M_y \cdot \Delta_{ry}\|
\]

\section{Example $M_x$, $M_y$}\label{example-deriv}
\noindent\textbf{Linear Regression: }
\[
\nabla\phi(x,y,\theta) = (\theta^Tx - y)x
\]
For a record, $r$, suppose we have a feature vector $x$.
If we take the partial derivatives with respect to x, $M_x$ is:
\[
M_x[i,i] = 2x[i] + \sum_{i \ne j} \theta[j]x[j] - y 
\]
\[
M_x[i,j] = \theta[j]x[i]
\]
Similarly $M_y$ is:
\[
M_y[i,1] = x[i] 
\]

\vspace{0.5em}

\noindent\textbf{Logistic Regression: } 
\[
\nabla\phi(x,y,\theta) = (h(\theta^Tx) - y)x
\]
where
\[
h(z) = \frac{1}{1+e^{-z}}
\]
we can rewrite this as:
\[
h_{\theta}(x) = \frac{1}{1+e^{\theta^Tx}}
\]
\[
\nabla\phi(x,y,\theta) = (h_{\theta}(x) - y)x
\]
In component form,
\[
g = \nabla\phi(x,y,\theta)
\]
\[
g[i] = h_{\theta}(x)\cdot x[i] - yx[i]
\]
Therefore,
\[
M_x[i,i] = h_{\theta}(x)\cdot(1- h_{\theta}(x))\cdot \theta[i] x[i] + h_{\theta}(x) - y
\]
\[
M_x[i,j] = h_{\theta}(x)\cdot(1- h_{\theta}(x))\cdot \theta[j] x[i] + h_{\theta}(x)
\]
\[
M_y[i,1] = x[i] 
\]

\noindent\textbf{SVM: } 
\[
\nabla\phi(x,y,\theta) =
\begin{cases}      
-y\cdot\boldsymbol{x} \text{ if } y\cdot\boldsymbol{x}\cdot\theta \le 1 \\
0\ \text{ if } y\ \boldsymbol{x}\cdot\theta \geq 1      
\end{cases}
\]
Therefore,
\[
M_x[i,i] = \begin{cases}      
-y[i] \text{ if } y\cdot\boldsymbol{x}\cdot\theta \le 1 \\
0\ \text{ if } y\ \boldsymbol{x}\cdot\theta \geq 1      
\end{cases} 
\]
\[
M_x[i,j] = 0
\]
\[
M_y[i,1] = x[i] 
\]

\section{Aggregate Queries as \\ Convex Losses}
\subsection{AVG and SUM queries}
\avgfunc, \sumfunc queries are a special case of the convex loss minimization discussed in the paper:
If we define the following loss, it is easy to verify the the optimal $\theta$ is the mean $\mu$:
\[
\phi = (x_{i} - \theta)^2
\]
with the appropriate scaling it can support $\avgfunc$, $\sumfunc$ queries with and without predicates.
Taking the gradient of that loss:
\[
\nabla\phi = 2(x_{i} - \theta)
\]
It is also easy to verify that the bound on errors is $O(\frac{\mathbb{E}((x-\mu)^2}{bT})$, which is essentially the CLT.
The importance sampling results are inutitive as well.
Applying the linearization:
\[
M_x = 2
\]
The importance sampling prioritizes points that it expects to be far away from the mean.

\subsection{MEDIAN}
Similarly, we can analyze the \medfunc query.
If we define the following loss, it is easy to verify the the optimal $\theta$ is the median $m$:
\[
\phi = \mid x_{i} - \theta\mid
\]
Taking the gradient of that loss:
\[
\nabla\phi = \text{1 if < m, -1 if > m}
\]
Applying the linearization:
\[
M_x = 0
\]
The intuitive result is that a robust query like a median does not need to consider estimation as the query result is robust to small changes.

\section{Experimental Comparison}
\subsection{Robust Logistic Regression}\label{rlogit}
We use the algorithm from Feng et al. for robust logistic regression.
\begin{enumerate}
\item Input: Contaminated training samples $\{(x_1, y_1), . . . ,(x_{n}
, y_{n})\}$ an upper bound on the number of outliers n, number of inliers n and sample dimension p.
\item Initialization: Set \[T = 4\sqrt{\log p/n + \log n/n}\]
\item Remove samples $(xi
, yi)$ whose magnitude satisfies $\|x_i\| \ge T$.
\item Solve regularized logistic regression problem.
\end{enumerate}

\section{Dollars For Docs Setup}\label{dfd-errors}
The dollars for docs dataset has the following schema:
\begin{lstlisting}[mathescape,basicstyle={\scriptsize}]
Contribution(pi_speciality$\textrm{,}$ drug_name$\textrm{,}$ device_name$\textrm{,}$
corporation$\textrm{,}$ amount$\textrm{,}$ dispute$\textrm{,}$ status)
\end{lstlisting}
To flag suspect donations, we used the \texttt{status} attribute.
When the \texttt{status} was ``covered" that means it was an allowed contribution under the researcher's declared protocol.
When the \texttt{status} was ``non-covered" that means it was a disallowed contribution under the researcher's declared protocol.
The rest of the textual attributes were featurized with a bag-of-words model, and the numerical amount and dispute attributes were treated as numbers.

We cleaned the entire Dollars for Docs dataset upfront to be able to evaluate how different budgeted data cleaning strategies compare to cleaning the full data.
To clean the dataset, we loaded the entire data 240089 records into Microsoft Excel. We identified four broad classes of errors:
\vspace{0.25em}

\noindent \textbf{Corporations are inconsistently represented: } ``Pfizer", ``Pfizer Inc.", ``Pfizer Incorporated".

\vspace{0.25em}

\noindent \textbf{Drugs are inconsistently represented: } ``TAXOTERE  DOCETAXEL -PROSTATE CANCER" and ``TAXOTERE"

\vspace{0.25em}

\noindent \textbf{Label of covered and not covered are not consistent: } ``No", ``Yes",``N", ``This study is not supported", ``None", ``Combination"

\vspace{0.25em} 

\noindent \textbf{Research subject must be a drug OR a medical device and not both: } ``BIO FLU QPAN H7N9AS03 Vaccine" and ``BIO FLU QPAN H7N9AS03 Device"

\vspace{0.5em} 

To fix these errors, we sorted by each column and merged values that looked similar and removed inconsistencies as in the status labels. 
When there were ambiguities, we refered to the drug company's website and whitepapers.
When possible, we used batch data transformations, like find and replace (i.e. the Set-of-Records model).
In all, 44234 records had some error and full data cleaning required about 2 days of efforts.

Once cleaned, in our experiment, we encoded the 4 problems as data quality constraints.
To fix the constraints, we looked up the clean value in the dataset that we cleaned up front.

\vspace{0.25em}

\noindent \textbf{Rule 1: } Matching dependency on corporation (Weighted Jaccard Similarity $>$ 0.8).

\vspace{0.25em}

\noindent \textbf{Rule 2: } Matching dependency on drug (Weighted Jaccard Similarity $>$ 0.8).

\vspace{0.25em}

\noindent \textbf{Rule 3: } Label must either be ``covered" or ``not covered".

\vspace{0.25em} 

\noindent \textbf{Rule 4: } Either drug or medical device should be null.

\vspace{0.5em}

\section{MNIST Setup}
We include visualization of the errors that we generated for the MNIST experiment.
We generated these errors in MATLAB by taking the grayscale version of the image (a $64\times 64$ matrix) and corrupting them by block removal and fuzzying.

\begin{figure}[ht]
\centering
\includegraphics[scale=0.20]{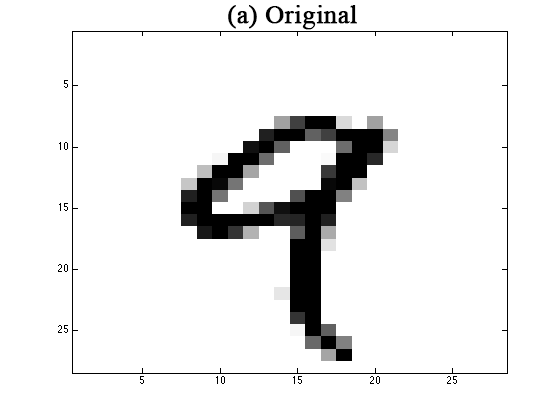}
 \includegraphics[scale=0.20]{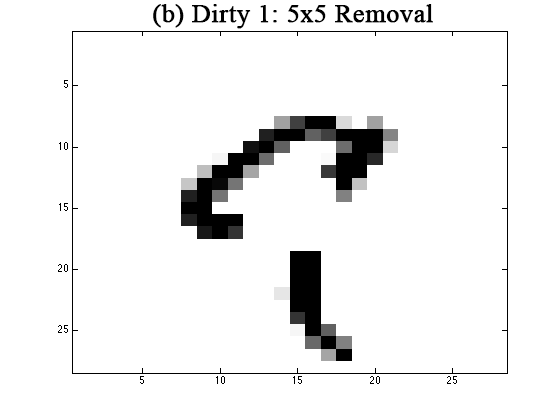}
 \includegraphics[scale=0.20]{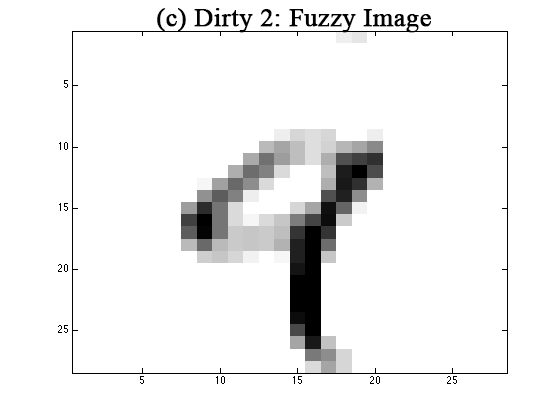}
 \caption{We experiment with two forms of corruption in the MNIST image datasets: 5x5 block removal and making the images fuzzy. Image (a) shows an uncorrupted ``9", image (b) shows one corrupted with block removal, and image (c) shows one that is corrupted with fuzziness. \label{mnist-corr}}
\end{figure}

%% file: outline.bbl
\begin{thebibliography}{10}

\bibitem{bdas}
Berkeley data analytics stack.
\newblock \url{https://amplab.cs.berkeley.edu/software/}.

\bibitem{dollarsfordocs}
Dollars for docs.
\newblock \url{http://projects.propublica.org/open-payments/}.

\bibitem{nytimes}
For big-data scientists, 'janitor work' is key hurdle to insights.
\newblock
  http://www.nytimes.com/2014/08/18/technology/for-big-data-scientists-hurdle-to-insights-is-janitor-work.html.

\bibitem{dollarsfordocsa}
A pharma payment a day keeps docs' finances okay.
\newblock
  \url{https://www.propublica.org/article/a-pharma-payment-a-day-keeps-docs-finances-ok}.

\bibitem{alexandrov2014stratosphere}
A.~Alexandrov, R.~Bergmann, S.~Ewen, J.~Freytag, F.~Hueske, A.~Heise, O.~Kao,
  M.~Leich, U.~Leser, V.~Markl, F.~Naumann, M.~Peters, A.~Rheinl{\"{a}}nder,
  M.~J. Sax, S.~Schelter, M.~H{\"{o}}ger, K.~Tzoumas, and D.~Warneke.
\newblock The stratosphere platform for big data analytics.
\newblock {\em {VLDB} J.}, 23(6), 2014.

\bibitem{altowim2014progressive}
Y.~Altowim, D.~V. Kalashnikov, and S.~Mehrotra.
\newblock Progressive approach to relational entity resolution.
\newblock {\em {PVLDB}}, 7(11), 2014.

\bibitem{DBLP:conf/sigmod/BergmanMNT15}
M.~Bergman, T.~Milo, S.~Novgorodov, and W.~C. Tan.
\newblock Query-oriented data cleaning with oracles.
\newblock In {\em SIGMOD Conference}, 2015.

\bibitem{bertsekas2011incremental}
D.~P. Bertsekas.
\newblock Incremental gradient, subgradient, and proximal methods for convex
  optimization: {A} survey.
\newblock {\em CoRR}, abs/1507.01030, 2015.

\bibitem{bottou2012stochastic}
L.~Bottou.
\newblock Stochastic gradient descent tricks.
\newblock In {\em Neural Networks: Tricks of the Trade - Second Edition}. 2012.

\bibitem{crotty2014tupleware}
A.~Crotty, A.~Galakatos, and T.~Kraska.
\newblock Tupleware: Distributed machine learning on small clusters.
\newblock {\em {IEEE} Data Eng. Bull.}, 37(3), 2014.

\bibitem{dekel2012optimal}
O.~Dekel, R.~Gilad{-}Bachrach, O.~Shamir, and L.~Xiao.
\newblock Optimal distributed online prediction using mini-batches.
\newblock {\em JMLR}, 13, 2012.

\bibitem{deshpande2004model}
A.~Deshpande, C.~Guestrin, S.~Madden, J.~M. Hellerstein, and W.~Hong.
\newblock Model-driven data acquisition in sensor networks.
\newblock In {\em VLDB}, 2004.

\bibitem{drineas2012fast}
P.~Drineas, M.~Magdon{-}Ismail, M.~W. Mahoney, and D.~P. Woodruff.
\newblock Fast approximation of matrix coherence and statistical leverage.
\newblock {\em JMLR}, 13, 2012.

\bibitem{feng2014robust}
J.~Feng, H.~Xu, S.~Mannor, and S.~Yan.
\newblock Robust logistic regression and classification.
\newblock In {\em NIPS}, 2014.

\bibitem{friedman2001elements}
J.~Friedman, T.~Hastie, and R.~Tibshirani.
\newblock {\em The elements of statistical learning}, volume~1.
\newblock Springer series in statistics Springer, Berlin, 2001.

\bibitem{gokhale2014corleone}
C.~Gokhale, S.~Das, A.~Doan, J.~F. Naughton, N.~Rampalli, J.~Shavlik, and
  X.~Zhu.
\newblock Corleone: Hands-off crowdsourcing for entity matching.
\newblock In {\em SIGMOD}, 2014.

\bibitem{gruenheid2014incremental}
A.~Gruenheid, X.~L. Dong, and D.~Srivastava.
\newblock Incremental record linkage.
\newblock {\em {PVLDB}}, 7(9), 2014.

\bibitem{guillory2009active}
A.~Guillory, E.~Chastain, and J.~Bilmes.
\newblock Active learning as non-convex optimization.
\newblock In {\em International Conference on Artificial Intelligence and
  Statistics}, 2009.

\bibitem{heise2014estimating}
A.~Heise, G.~Kasneci, and F.~Naumann.
\newblock Estimating the number and sizes of fuzzy-duplicate clusters.
\newblock In {\em CIKM Conference}, 2014.

\bibitem{tensor}
G.~Inc.
\newblock Tensorflow.
\newblock \url{https://www.tensorflow.org/}.

\bibitem{DBLP:conf/pervasive/JefferyAFHW06}
S.~R. Jeffery, G.~Alonso, M.~J. Franklin, W.~Hong, and J.~Widom.
\newblock Declarative support for sensor data cleaning.
\newblock In {\em Pervasive Computing}, 2006.

\bibitem{kandel2012}
S.~Kandel, A.~Paepcke, J.~M. Hellerstein, and J.~Heer.
\newblock Enterprise data analysis and visualization: An interview study.
\newblock {\em {IEEE} Trans. Vis. Comput. Graph.}, 18(12), 2012.

\bibitem{DBLP:conf/recsys/KrishnanPFG14}
S.~Krishnan, J.~Patel, M.~J. Franklin, and K.~Goldberg.
\newblock A methodology for learning, analyzing, and mitigating social
  influence bias in recommender systems.
\newblock In {\em RecSys}, 2014.

\bibitem{krishnan2015svc}
S.~Krishnan, J.~Wang, M.~J. Franklin, K.~Goldberg, and T.~Kraska.
\newblock Stale view cleaning: Getting fresh answers from stale materialized
  views.
\newblock {\em {PVLDB}}, 8(12), 2015.

\bibitem{DBLP:conf/case/MahlerKLSMKPWFAG14}
J.~Mahler, S.~Krishnan, M.~Laskey, S.~Sen, A.~Murali, B.~Kehoe, S.~Patil,
  J.~Wang, M.~Franklin, P.~Abbeel, and K.~Y. Goldberg.
\newblock Learning accurate kinematic control of cable-driven surgical robots
  using data cleaning and gaussian process regression.
\newblock In {\em CASE}, 2014.

\bibitem{mayfield2010eracer}
C.~Mayfield, J.~Neville, and S.~Prabhakar.
\newblock {ERACER:} a database approach for statistical inference and data
  cleaning.
\newblock In {\em SIGMOD Conference}, 2010.

\bibitem{DBLP:journals/pvldb/MozafariSFJM14}
B.~Mozafari, P.~Sarkar, M.~J. Franklin, M.~I. Jordan, and S.~Madden.
\newblock Scaling up crowd-sourcing to very large datasets: {A} case for active
  learning.
\newblock {\em {PVLDB}}, 8(2), 2014.

\bibitem{nelson2012query}
B.~Nelson, B.~I.~P. Rubinstein, L.~Huang, A.~D. Joseph, S.~J. Lee, S.~Rao, and
  J.~D. Tygar.
\newblock Query strategies for evading convex-inducing classifiers.
\newblock {\em JMLR}, 13, 2012.

\bibitem{pan2010survey}
S.~J. Pan and Q.~Yang.
\newblock A survey on transfer learning.
\newblock {\em TKDE}, 22(10), 2010.

\bibitem{papenbrock2015progressive}
T.~Papenbrock, A.~Heise, and F.~Naumann.
\newblock Progressive duplicate detection.
\newblock {\em {IEEE} Trans. Knowl. Data Eng.}, 27(5), 2015.

\bibitem{settles2010active}
B.~Settles.
\newblock Active learning literature survey.
\newblock {\em University of Wisconsin, Madison}, 52:11, 2010.

\bibitem{simpson1951interpretation}
E.~H. Simpson.
\newblock The interpretation of interaction in contingency tables.
\newblock {\em Journal of the Royal Statistical Society. Series B
  (Methodological)}, 1951.

\bibitem{wang1999sample}
J.~Wang, S.~Krishnan, M.~J. Franklin, K.~Goldberg, T.~Kraska, and T.~Milo.
\newblock A sample-and-clean framework for fast and accurate query processing
  on dirty data.
\newblock In {\em SIGMOD Conference}, 2014.

\bibitem{whang2014incremental}
S.~E. Whang and H.~Garcia{-}Molina.
\newblock Incremental entity resolution on rules and data.
\newblock {\em {VLDB} J.}, 23(1), 2014.

\bibitem{simpsonsparadox}
C.~Woolston.
\newblock Gender-disparity study faces attack.
\newblock
  \url{http://www.nature.com/news/gender-disparity-study-faces-attack-1.18428}.

\bibitem{xiaofeature}
H.~Xiao, B.~Biggio, G.~Brown, G.~Fumera, C.~Eckert, and F.~Roli.
\newblock Is feature selection secure against training data poisoning?
\newblock In {\em ICML}, 2015.

\bibitem{yakout2013don}
M.~Yakout, L.~Berti{-}Equille, and A.~K. Elmagarmid.
\newblock Don't be scared: use scalable automatic repairing with maximal
  likelihood and bounded changes.
\newblock In {\em SIGMOD Conference}, 2013.

\bibitem{DBLP:journals/pvldb/YakoutENOI11}
M.~Yakout, A.~K. Elmagarmid, J.~Neville, M.~Ouzzani, and I.~F. Ilyas.
\newblock Guided data repair.
\newblock {\em {PVLDB}}, 4(5), 2011.

\bibitem{zhao2014stochastic}
P.~Zhao and T.~Zhang.
\newblock Stochastic optimization with importance sampling for regularized loss
  minimization.
\newblock In {\em ICML}, 2015.

\end{thebibliography}
